\documentclass[a4, twocolumn, 10pt]{IEEEtran}

\usepackage{graphicx} 
\usepackage{subfigure}
\usepackage{amsmath}
\usepackage{amssymb}
\usepackage{amsthm}
\usepackage{bm}
\usepackage{color}
\usepackage{url}
\usepackage{cite}
\usepackage{epstopdf}

\graphicspath{{./fig/}}

\newcommand{\undersmile}{\mathrel{\lower6pt\hbox{$\smile$}}}
\renewcommand{\baselinestretch}{1}

\newcommand{\qh}{{\bf h}}

\newcommand{\qn}{{\bf n}}

\newcommand{\qr}{{\bf r}}

\newcommand{\qu}{{\bf u}}

\newcommand{\qw}{{\bf w}}

\newcommand{\qy}{{\bf y}}

\newcommand{\qB}{{\bf B}}

\newcommand{\qH}{{\bf H}}
\newcommand{\qI}{{\bf I}}

\newcommand{\qR}{{\bf R}}

\newcommand{\qU}{{\bf U}}
\newcommand{\qV}{{\bf V}}
\newcommand{\qW}{{\bf W}}

\newcommand{\qY}{{\bf Y}}
\newcommand{\qZ}{{\bf Z}}

\newcommand{\be}{\begin{equation}} \newcommand{\ee}{\end{equation}}
\newcommand{\bea}{\begin{eqnarray}} \newcommand{\eea}{\end{eqnarray}}

\newtheorem{lemma}{Lemma}
\newtheorem{proposition}{Proposition}

\begin{document}
\title{Beamforming Optimization for Full-Duplex Wireless-powered MIMO Systems}
\author{\IEEEauthorblockN{Batu K. Chalise \IEEEmembership{Senior Member, IEEE}, Himal A. Suraweera \IEEEmembership{Senior Member, IEEE}, Gan Zheng \IEEEmembership{Senior Member, IEEE}, George K. Karagiannidis \IEEEmembership{Fellow, IEEE}}
\thanks{B. K. Chalise is with the Department of Electrical Engineering and Computer Science, Cleveland State University, 2121 Euclid Avenue, OH 44115, USA (e-mail: b.chalise@csuohio.edu).}
\thanks{H. A. Suraweera is with the Department of Electrical and Electronic Engineering, University of Peradeniya, Peradeniya 20400, Sri Lanka (e-mail: himal@ee.pdn.ac.lk).}
\thanks{G. Zheng is with the Wolfson School of Mechanical, Electrical, and Manufacturing Engineering, Loughborough University, Loughborough, LE11 3TU, UK (e-mail: g.zheng@lboro.ac.uk).}
\thanks{G. K. Karagiannidis is with Electrical and Computer Engineering Department, Aristotle University of Thessaloniki, 54636 Thessaloniki, Greece (e-mail: geokarag@auth.gr).}
\thanks{Part of this work was presented in IEEE SPAWC'16, Edinburgh, UK, July 2016.}  
}
\maketitle


\begin{abstract}

We propose techniques for optimizing transmit beamforming in a full-duplex multiple-input-multiple-output (MIMO) wireless-powered communication system, which consists of two phases. In the first phase, the wireless-powered mobile station (MS) harvests energy using signals from the base station (BS), whereas in the second phase, both MS and BS communicate to each other in a full-duplex mode. When complete instantaneous channel state information (CSI)  is available, the BS beamformer and the time-splitting (TS) parameter of energy harvesting are jointly optimized in order to obtain the BS-MS rate region. The joint optimization problem is non-convex, however, a computationally efficient optimum technique, based upon semidefinite relaxation and line-search,  is proposed to solve the problem.  A sub-optimum zero-forcing approach is also proposed, in which a closed-form solution of TS parameter is obtained.  When only second-order statistics of transmit CSI is available, we propose to maximize the ergodic information rate at the MS, while maintaining the outage probability at the BS below a certain threshold. An upper bound for the outage probability is also derived  and an approximate convex optimization framework is proposed for efficiently solving the underlying non-convex  problem. Simulations demonstrate the advantages of the proposed methods over the sub-optimum and half-duplex ones. 

\end{abstract}


\section{Introduction}
{\large{P}}roliferation of communication devices, systems and networks has  considerably increased the demand for wireless spectrum, driving the interest to design systems with higher spectral efficiency. Most contemporary bi-directional wireless systems have been developed for half-duplex (HD) operation (i.e., either transmit or receive, but not both simultaneously). As an effective method of improving the spectral efficiency of contemporary HD systems, full-duplex (FD) communications have emerged as a promising solution \cite{Sabharwal-JSAC14, Hanzo}. Although the concept of FD is not new and has been in use since 1940s, so far it has been considered as impossible to realize due to the loopback interference (LI) that couples the device output to the input \cite{Choi-MOBICOM10,Riihonen-TWC11}. However, FD is now becoming feasible, thanks to promising analog/digital and spatial domain LI cancellation techniques that can achieve high transmit-receive isolation \cite{Riihonen-TSP11, Duarte-THESIS, Korpi_AntP16,Ngo-JSAC14}. As a result, experimental demonstration of the feasibility of FD has already been carried out by several research laboratories. 

FD communications can be implemented for three basic topologies, namely, (a) relay topology (b) bidirectional topology and (c) base station (BS) topology \cite{Sabharwal-JSAC14}. To this end, bidirectional FD systems have been investigated in some existing works in the literature \cite{Day-TSP12,Cirik_TSP14,Senaratne-ICC11, Ju-TWC11}. These include papers that have focused on information-communication theoretic performance metrics, such as the achievable sum rate and symbol error probability. In \cite{Day-TSP12}, achievable upper and lower sum-rate bounds of multiple antenna bidirectional communication, that use pilot-aided channel estimates for transmit/receive beamforming and interference cancellation, were derived. The beamforming performance of bidirectional multiple-input multiple-output (MIMO) transmission with spatial LI mitigation was investigated in \cite{Senaratne-ICC11}. Furthermore, capacity of a bidirectional MIMO system with spatial correlation was presented in \cite{Ju-TWC11}. Finally, the maximization of the asymptotic ergodic mutual information for a MIMO bi-directional communication system, with imperfect channel state information (CSI), was the focus of the work in \cite{Cirik_TSP14}.

In addition to the spectral efficiency, energy efficiency has gained wide research attention for the design of wireless networks. For example, energy constraints impose an upper limit on the transmit power and the associated signal processing in wireless devices. To this end, a new paradigm that can power communication devices via energy harvesting techniques has emerged \cite{Lu-COM_SURV15,Bi-COMMAG16}. Among different energy harvesting sources such as ambient heat, wind, solar, vibration, etc., wireless power transfer (WPT) using dedicated radio frequency sources is regarded as a promising solution, since it can be controlled to achieve optimum performance. Therefore, WPT can be used to remotely power a variety of applications such as wireless sensor networks, body area networks, wireless charging facilities and future cellular networks \cite{Lu-COM_SURV15,Bi-COMMAG16,HUANG-TWC14}. Moreover, since wireless signals can transport both information and energy, by introducing the new notion of ``simultaneous wireless information and power transfer'' (SWIPT), the rate-energy region of a wireless-powered MIMO broadcast network with an external energy harvester was characterized in \cite{Zhang-TWC13}. In \cite{Huang-TVT16}, throughput performance of a wireless-powered network in which a multi-antenna hybrid access point (H-AP) beamforms energy to a single antenna user in order to assist uplink information transfer was presented. Motivated by the advantages of FD and WPT, some recent works have also investigated the performance of wireless-powered bidirectional communications \cite{Ju-TCOM14,Yamazaki-WCNC15,Hu-ARXIV15,Gao-WCNC15}. 

In wireless-powered FD networks, the deployment of multiple antennas can be considered as a practical solution, since the strategy is useful to harvest higher amount of energy \cite{Ding-COMM15} as well as to deploy spatial beamforming techniques to suppress LI \cite{Riihonen-TSP11}. In \cite{Ju-TCOM14}, considering a FD H-AP that broadcasts wireless energy to a set of downlink users while receiving information from a set of uplink users, a solution to an optimal resource allocation problem was presented. In \cite{Yamazaki-WCNC15}, hardware implementation of a wireless system, that transmits data and power in the same frequency, was presented. In \cite{Gao-WCNC15}, performance of a wireless-powered FD communication network, which consists of a dual-antenna FD H-AP and a single dual-antenna FD user, was investigated. Specifically, assuming different roles for the two antennas (downlink WPT or uplink wireless information transfer), closed-form expressions for the system's outage probability and the ergodic capacity were derived. More recently, in \cite{Hu-ARXIV15}, a weighted sum transmit power optimization problem for a bidirectional FD system with WPT was formulated and solved. However, it assumes perfect LI cancellation at terminals which is impossible in practice \cite{Riihonen-TWC11}.

Inspired by wireless-powered FD communications, in this paper, we consider bidirectional communication between an $N$-antenna BS and a mobile station (MS) with two antennas. According to ``harvest-then-transmit'' protocol \cite{ChaliseSPAWC2016}, the BS first transmits energy to the MS, which is used by the MS for the subsequent uplink transmission. At the end of the energy transfer phase, both BS and MS simultaneously transfer information in the uplink and downlink, thanks due to the FD operation. Specifically for this setup, we propose methods for jointly optimizing the beamformer at the BS and the time-splitting (TS) parameter that divides a given time-slot into energy harvesting and  data transmission phases. Both full and partial CSI cases are considered. In the former case, where the instantaneous channel is known, the optimized boundary of the BS-MS rate region is obtained, which describes the trade-off between BS and MS information rate. To this end, a computationally efficient optimum method,  based upon semidefinite relaxation (SDR) and line-search,  is proposed and its performance is compared with a sub-optimum method that uses the zero-forcing (ZF) criterion for designing the beamformer. 
 
In the partial CSI case, the BS and MS know only the second-order statistics, such as channel covariance matrices,  of their transmit CSI. It is also worth mentioning that the BS rate turns to be much smaller than the MS rate, since the MS transmits with the harvested energy which, in general, is much smaller than the transmit power of the BS. Moreover, the maximum possible value of the BS rate cannot be achieved in the partial CSI case. Due to these reasons, it is important to ensure that the BS is not in outage rather than maximize the BS-rate which is already constrained by the MS's transmit power. Hence, we propose to maximize the ergodic information rate at the MS, while ensuring that the outage probability at the BS remains below a certain threshold value. This optimization problem  is non-convex and non-tractable. As such, we derive an upper bound of the BS outage probability, and formulate an optimization problem so that the gap between the derived upper bound and the exact outage  probability remains minimum. In particular, using the upper bound of the outage probability, we maximize the ergodic information rate at the MS. We utilize the monotonicity property of the derived  exact ergodic information rate and formulate an SDR optimization problem, that is efficiently solved with a convex optimization toolbox. 
 
The main contributions of this paper are summarized as follows:
\begin{itemize}
\item In the case of full CSI,  the joint optimization problem of transmit beamforming and TS parameter is efficiently solved as an SDR problem. The optimality of the relaxation is confirmed with a proof that the optimum solution of the relaxed problem is rank-one.  
\item We show that the MS-rate is a  monotonically decreasing function of the TS parameter, and this property is utilized to efficiently solve the SDR-based joint optimization. A closed-form expression of the TS parameter is derived in the ZF-based sub-optimum design.  
\item  Closed-form expressions for the ergodic MS-rate and BS outage probability are derived. For a given TS parameter, we show that the ergodic MS-rate is a monotonically increasing function of the beamformer gain towards the MS. We, then,  utilize this property for solving the problem of maximizing the MS ergodic rate, while satisfying the outage probability constraint at the BS. 
\item Since the optimization problem remains non-tractable with the original outage probability, we derive its upper bound and approximate the original optimization problem with the SDR problem. The proposed optimization tries to minimize the gap between the exact outage probability and its upper bound.    
\end{itemize}

The rest of the paper is organized as follows. The system model and problem formulation are presented in Section \ref{sec:Section1}.  The optimization problems for full and partial CSI cases are solved in sections  \ref{sec:Section2}  and  \ref{sec:Section3}, respectively. In Section \ref{sec:Section4}, numerical results are provided, whereas in Section \ref{sec:Section5}, conclusions are drawn. 

\hspace*{0.3cm} {\it Notation:} Upper (lower) bold face letters will be used for matrices (vectors); $(\cdot)^T$, $(\cdot)^{H}$, ${\rm E}\left\{\cdot\right\}$, ${\bf I}$ and  $||{\cdot}||$ denote transpose, Hermitian transpose, expectation w.r.t. a random variable $x$, identity matrix, and Frobenius norm (Euclidean norm for a vector),  respectively. ${\rm tr} (\cdot)$, ${\mathcal C}^{M\times M}$,  and ${\bf A}\succeq 0$ denote the matrix trace operator,  space of $M\times M$ matrices with complex entries, and positive semidefiniteness of ${\bf A}$, respectively.

\section{System Model and Problem Formulation }
\label{sec:Section1}
We consider bidirectional FD communications  between an $N$-antenna BS and a MS as shown in Fig. \ref{fig-fig1}. Specifically, the BS has $N_t$ transmit antennas and $N_r\triangleq N-N_t$ receive antennas. Notice that, $N_t$, together with the chosen transmit/receive antennas could be optimized, but we keep them fixed. Although joint antenna selection and beamformer optimization is an interesting future work, it requires a multi-stage optimization approach and, thus, is is not considered in this work. The MS is an energy constrained device and harvests energy from the signals transmitted by the  BS. The MS, then, utilizes the harvested energy for its uplink transmission. Since the MS is energy constrained and depends on the harvested energy (which assumes typically small values), as in \cite{Hu-ARXIV15} we assume that the MS is equipped with two antennas (one antenna of the MS is used for transmission, whereas the other is used for reception). This assumption is further motivated by the fact that the space constraint prevents mounting more antennas at the MS. Moreover, since the BS is equipped with multiple antennas in our system, for sufficient amounts of energy harvesting, beamforming can be effectively used \cite{Zhang-TWC13}.

\begin{figure*}[htb!]
  \centering
  \subfigure[ EH phase of duration $\alpha T$.]{\includegraphics[scale=0.6]{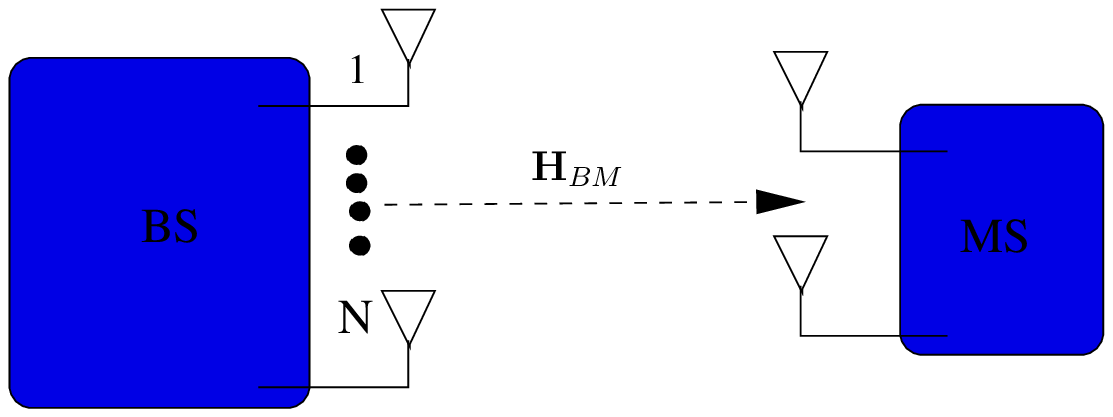}}\quad
  \subfigure[ FD communication phase of duration $(1-\alpha)T$.]{\includegraphics[scale=0.4]{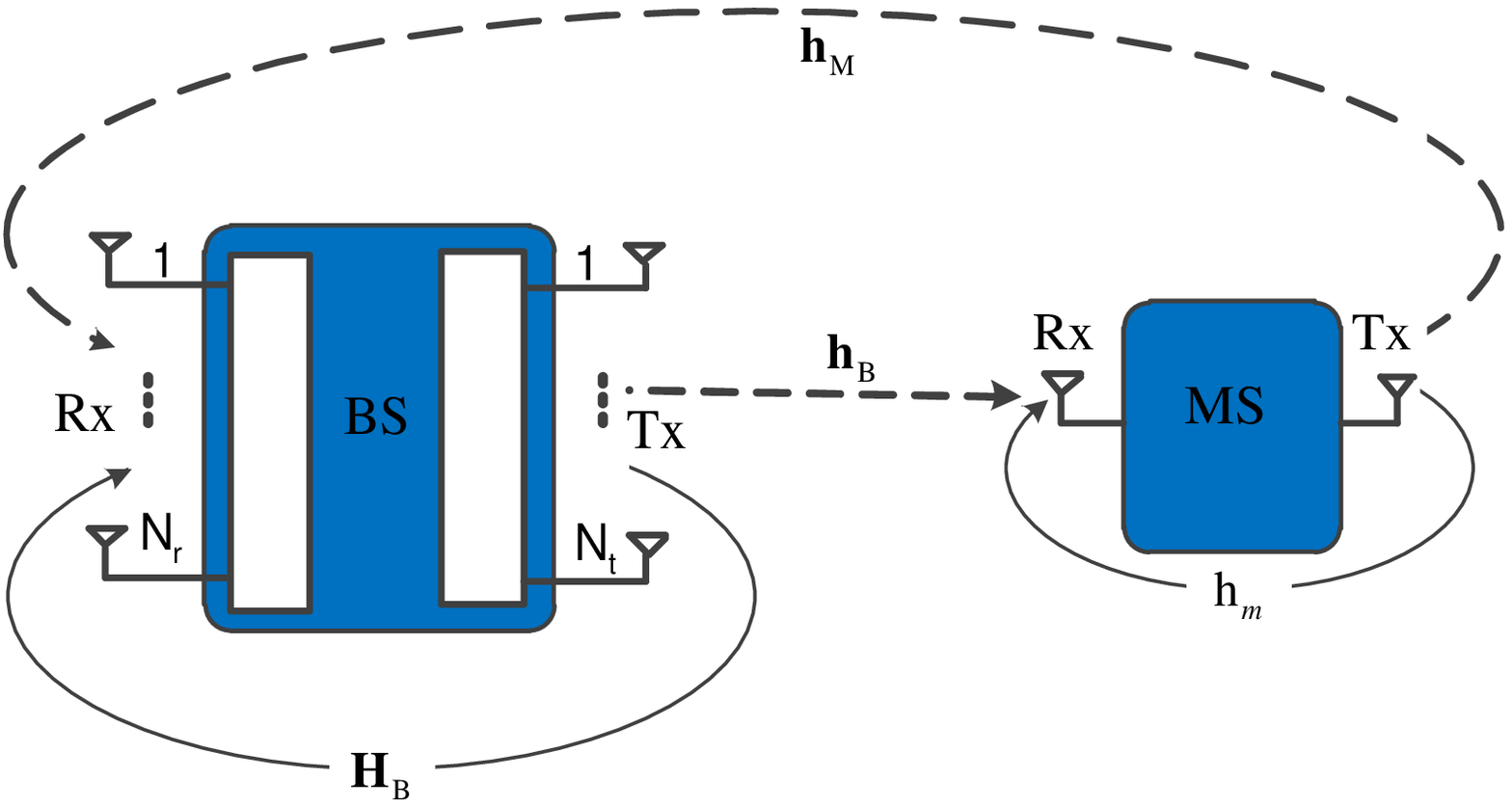}}
  \caption{Two phases of wireless-powered communication system.}
  \label{fig-fig1}
\end{figure*}

Without loss of generality, assuming a block time of $T=1$, communication between the BS and MS takes place in two phases with duration $\alpha$ and $(1-\alpha)$, respectively. In phase I, the BS employs all of its antennas to transmit energy, whereas the MS employs both of its antennas for reception. The signal received by the MS during energy harvesting phase is given by ${\bf y}_E={\bf H}_{BM}{\bf w}_E s_E+{\bf n}_E$, where  $\qH_{BM}\in {\mathcal C}^{2\times N}$ is the channel between the BS and the MS,  ${\bf w}_E \in {\mathcal C}^{N\times 1}$ is the energy beamformer, $s_E$ is the signal transmitted by the BS, and ${\bf n}_E$ is the additive white Gaussian noise (AWGN) at the MS. We assume that the harvested energy due to the noise (including both the antenna noise and the rectifier noise) is small and thus ignored \cite{MohammadiTC2016}. Thus, assuming that ${\rm E}\{|s_E|^2\}=1$ during the period of $\alpha$, the harvested energy can be expressed as $E=\eta \alpha || {\bf H}_{BM}{\bf w}_E||^2=\eta \alpha {\rm tr}\left( {\bf w}^H_E  {\bf H}^H_{BM} {\bf H}_{BM}  {\bf w}_E\right)$, where $\eta$ is the conversion efficiency of the rectifier circuit at the MS. Considering that $|| {\bf w}_E||^2=P$, where $P$ is the total transmit power of the BS during energy harvesting phase, it is clear that the optimum   ${\bf w}_E$ is given by ${\bf w}_E=\sqrt{P}{\bf v}_{\rm max}$, where ${\bf v}_{\rm max}$ is the eigenvector corresponding to the largest eigenvalue of the matrix  ${\bf H}^H_{BM} {\bf H}_{BM}$. This means, the harvested energy is given by
 \be
    E = P \alpha{\bar \lambda}{ \left(\qH_{BM}\qH_{BM}^H\right)},
 \ee
where the channel between the BS and the MS is denoted as $\qH_{BM}$ and ${\bar \lambda}(\cdot)$ returns the maximum eigenvalue of a matrix. As such, the BS requires transmit CSI which is obtained through {\it reverse-link training via channel reciprocity}\cite{BrunoRui2017} approach. More specifically, assuming that the BS-MS and MS-BS channels are reciprocal, the MS first sends training{\footnote{Note that the MS will not be completely operated with the harvested RF energy. The energy from a battery can be used to support most critical and basic functions, such as switching on/off of  transceiver circuits, sending control and training signals, etc. This assumption is standard in the wireless energy harvesting communications  literature (see for e.g., \cite{BrunoRui2017} and the references therein).} and the BS then estimates the channel and performs optimum energy beamforming. Although it requires CSI at the BS, the harvested energy due to beamforming gain can be much larger than the energy consumed for sending training signals \cite{BrunoRui2017}. On the other hand, we will relax the full CSI requirement by considering partial CSI case and solving the corresponding optimization problem in Section \ref{sec:Section3}. Thus, the reported results of the full CSI case serve as useful theoretical bounds for practical design.

Note that, in (1) the energy that can be harvested from noise is omitted since for all practical purposes it is negligible. $\qH_{BM}$ is expressed as $\qH_{BM}=\sqrt{\frac{1}{d^\tau}}{\bar \qH}_{BM}$, where $d$ is the distance between the BS and MS, $\tau$ is the path loss exponent, and each element of ${\bar \qH}_{BM}$ has zero-mean and unit-variance. In phase II, since both terminals operate in the FD mode, the BS and MS simultaneously communicate with each other. The transmit power of MS can be written as
 \be
    p_m =  \frac{ \alpha\eta P{\bar \lambda}{ \left(\qH_{BM}\qH_{BM}^H\right)} }{(1-\alpha)},
 \ee
 where $\eta$ is the conversion efficiency of wireless energy transfer. Let the $1\times N_t$ $BS\rightarrow MS$ channel be $\qh_{B}^H=\sqrt{\frac{1}{d^\tau}}{\bar \qh}_{B}^H$ and the $(N-N_t)\times 1$ $MS\rightarrow BS$ channel be $\qh_{M}=\sqrt{\frac{1}{d^\tau}}{\bar \qh}_M$, where all elements of ${\bar \qh}_{B}$ and ${\bar \qh}_{M}$ have zero-mean and unit variance. The residual LI channels are $\qH_B \in \mathbb{C}^{(N-N_t)\times N_t}$ and $h_m$ at the BS and MS, respectively. In order to reduce the deleterious effects of LI on system performance, we assume that an analog/digital cancellation scheme can be employed at the BS and MS, respectively and as such the residual channels are modeled as feedback fading channels \cite{Riihonen-TWC11, Riihonen-TSP11}. Since such a cancellation scheme can be characterized by a specific residual power, each element of ${\bf H}_B$, and $h_m$ can be modeled as zero-mean circularly symmetric complex Gaussian (ZMCSCG) random variables of variances $\sigma^2_{h_b}$ and $\sigma^2_{h_m}$, respectively. Modeling of the residual LI channel in such a way is now common and a standard assumption in the FD literature since the dominant line-of-sight component in LI can be removed effectively when a cancellation method is implemented \cite{Duarte-THESIS}. It is also important to emphasize that perfect cancellation of LI is not possible due to imperfect estimation of LI channel, inevitable transceiver chain impairments  \cite{Riihonen-TSP11}, \cite{Masmoudi}, and inherent processing delay. Therefore, ${\bf H}_B$ and $h_m$ can assume relatively large values and their effects can be minimized with spatial suppression techniques \cite{MohammadiTC2016}. 

\subsection{Signal Model}
The received signals at the BS and MS, are, respectively 
 \bea
 \label{eq:sigmod1}
    \qy_B =  \sqrt{p_m}\qh_{M} s_M + \qH_B \qw_B s_B + \qn_B,\nonumber\\
    y_M =  \qh_B^H \qw_B s_B +  \sqrt{p_m} h_m s_M + n_M,
 \eea
where $\qy_B\in{\mathcal C}^{(N-N_t) \times 1}$, $s_M$ and $s_B$ are the information symbols transmitted by the MS and BS, respectively, $\qw_B \in {\mathcal C}^{N_t\times 1} $ is the beamformer at the BS, $\qn_B \in {\mathcal C}^{(N-N_t) \times 1}$  is the additive White Gaussian noise (AWGN) vector at the receive antenna elements of the BS, and $n_M \in {\mathcal C}$ is the AWGN at the receive antenna of the MS. Furthermore, it is assumed that   ${\rm E}\left\{ s_B  \right\}={\rm E}\left\{ s_M  \right\}=0$,  ${\rm E}\left\{ |s_B|^2  \right\}={\rm E}\left\{ |s_M|^2  \right\}=1$, ${\rm E}\left\{  \qn_B\right\}={\bf 0}$,   ${\rm E}\left\{  n_M\right\}=0$, ${\rm E}\left\{  \qn_B \qn_B^H\right\}=\sigma_b^2{\bf I}_{(N-N_t)}$,   ${\rm E}\left\{  |n_M|^2\right\}=\sigma_m^2$, and signals and noise are statistically independent. We also consider that $||{\bf w}_B||^2=P$, which means that the BS transmits with the same power, $P$,  during both energy harvesting and communication phases. It is worthwhile to note that the BS can transmit different powers in two phases and the system performance can be further improved by optimizing these powers. However, this leads to a new optimization problem which is beyond the scope of the current work.   

\hspace*{0.3cm} The BS applies a beamformer $\qr_B\in{\mathcal C}^{(N-N_t)\times 1}$ to the received signal $\qy_B$. Without loss of generality, it is assumed that $||{\bf r}_B||=1$. The output after beamforming is given by
\bea
y_B=\qr^H_B\qy_B=\qr^H_B\left( \sqrt{p_m}\qh_{M} s_M + \qH_B \qw_B s_B + \qn_B\right). 
\eea
The signal-to-interference-and-noise ratio (SINR) at the BS and MS are then given by
\bea
SINR_B &= & \frac{p_m |\qr_B^H  \qh_{M}|^2}{\sigma_b^2+|\qr_B^H\qH_B\qw_B|^2},
\eea
and
\bea
SINR_M&=& \frac{|\qh_B^H\qw_B|^2}{\sigma_m^2+ p_m |h_m|^2},
\eea
respectively. For a given $\qw_B$, the optimum $\qr_B$ is the one that maximizes $SINR_B$ and, thus,  is obtained by solving
\bea
\label{eq:RxBF1}
\max\limits_{||\qr_B||=1} \frac{ \qr_B^H  \qh_{M} \qh^H_{M} \qr_B}{ \qr_B^H\left( \sigma_b^2\qI+\qH_B\qw_B\qw_B^H \qH_B^H\right)\qr_B},
\eea
which is in generalized Rayleigh quotient form \cite{Horn}. It is well known that the maximum value in (\ref{eq:RxBF1}) is obtained when
\bea
\label{eq:RxBF2}
\qr_B=\frac{ \left( \sigma_b^2\qI+\qH_B\qw_B\qw_B^H \qH_B^H\right)^{-1} \qh_M}{  || \left( \sigma_b^2\qI+\qH_B\qw_B\qw_B^H \qH_B^H\right)^{-1} \qh_M  || }.
\eea
Substituting the optimum $\qr_B$ into $SINR_B$, it is clear that 
\bea
\label{eq:RxBF3}
SINR_B=p_m\qh_M^H(\sigma_b^2\qI+  \qH_B\qw_B\qw_B^H\qH_B^H )^{-1}\qh_M.
\eea
Consequently, the BS achievable rate is given by
\bea
r_B=(1-\alpha)\log_2\left( 1 +  p_m \qh_M^H(\sigma_b^2\qI+  \qH_B\qw_B\qw_B^H\qH_B^H )^{-1}\qh_M\right),\nonumber
\eea
which with the help of Sherman-Morrison formula \cite{Hager} can also be expressed as
\bea
r_B \hspace*{-0.3cm} &= \hspace*{-0.3cm} & (1-\alpha)\log_2\left( 1 + \frac{ p_m}{\sigma_b^2} \left(\|\qh_M\|^2  -
\frac{|\qh_M^H\qH_B\qw_B |^2  }{\sigma_b^2+
\|\qH_B\qw_B\|^2} \right)\right). \nonumber
\eea
On the other hand, the achievable rate at the MS can be written as
 \bea
   r_M \hspace*{-0.3cm} &=\hspace*{-0.3cm} & (1-\alpha)\log_2\left( 1 + \frac{|\qh_B^H\qw_B|^2}{\sigma_m^2+ p_m |h_m|^2}\right).
 \eea

\subsection{Problem Formulation}
 Our objective is to characterize the bidirectional communications with the MS-BS rate region. It can be obtained by maximizing the MS rate while ensuring that the BS-rate is equal to a certain value, $R_B$. By solving this optimization problem for all  $R_B$, where $R_B\in [0, R_B^{\rm max}]$ and  $R_B^{\rm max}$ is the maximum value of BS rate, we obtain the MS-BS rate region. Note that  $R_B^{\rm max}$ is obtained by solving $R_B^{\rm max}=\max\limits_{||\qw_B||=1} r_B$. The closed-form expression of   $R_B^{\rm max}$  is derived in Appendix A. As such, the optimization problem for a given $R_B$ is expressed as
\bea
 \label{eq:ProbForm1}
     \max_{ \{\qw_B, \alpha\} }   \hspace*{-0.3cm}& & \hspace*{-0.3cm} ~(1-\alpha)\log_2\left( 1 + \frac{|\qh_B^H\qw_B|^2}{\sigma_m^2+ p_m
    |h_m|^2}\right)\nonumber\\
\hspace*{-0.3cm}     \mbox{s.t.}~~  \hspace*{-0.3cm}& & \hspace*{-0.3cm}   (1-\alpha)\log_2\biggl( 1 +  \frac{p_m}{\sigma_b^2} \|\qh_M\|^2  - \biggr. \nonumber \\ & &\hspace*{1.5cm}  \biggl.  \frac{\frac{p_m}{\sigma_b^2} |\qh_M^H\qH_B\qw_B |^2  }{\sigma_b^2+ \|\qH_B\qw_B\|^2}  \biggr) =R_B,\nonumber\\
          \hspace*{-0.3cm}& & \hspace*{-0.3cm} || \qw_B||^2 \leq P, ~~0<\alpha<1,\nonumber\\
           \hspace*{-0.3cm}& & \hspace*{-0.3cm} p_m = \frac{ \alpha\eta P{\bar \lambda}{ \left(\qH_{BM}\qH_{BM}^H\right)}  }{(1-\alpha)}.
         \eea
The optimization problem in \eqref{eq:ProbForm1} is a complicated non-convex problem, w.r.t. $\qw_B$ and $\alpha$. However, it can be efficiently solved by finding the optimum $\qw_B$ for a given $\alpha$ and vice-versa. Since $\alpha$ is a scalar, the optimum solution can be ascertained by using one-dimensional search, w.r.t. $\alpha$. 

 \section{Joint Optimization with Perfect CSI}
\label{sec:Section2}
In this section, we propose the optimum and sub-optimum methods for solving the joint optimization of beamformer and the TS  parameter in (\ref{eq:ProbForm1}), when perfect CSI is available. In practice, CSI is subject to different types of errors (e.g., channel estimation, feedback,  delay and quantization errors)  and, hence, the joint optimization with the assumption of perfect CSI  enables us to obtain the outer  boundary of the BS-MS rate region. Such a boundary provides an upper bound performance that can be achieved by a system that is subject to erroneous CSI.  
\subsection{Optimum Method}
In this method, the jointly optimal $\qw_B$  and $\alpha$ are determined by obtaining the optimum $\qw_B$  for each  $\alpha$, and then choosing  those  $\qw_B$  and $\alpha$  that maximize the objective function in  (\ref{eq:ProbForm1}). For this purpose, a grid-search over $\alpha$ is required, which is just one-dimensional (or linear). Exploiting the structure of the problem (\ref{eq:ProbForm1}), we show that the required  grid search  can be limited to a small region of $\alpha$, and hence, the computational cost for solving the joint optimization is minimized. 
\subsubsection{ Optimization of $\qw_B$}
   We first consider a problem to optimize $\qw_B$ for a given $\alpha$. In this case, the optimization problem  in \eqref{eq:ProbForm1} is expressed as
   \bea
  \label{eq:ProbForm2}
   \max_{ \qw_B }~   \hspace*{-0.3cm}& & \hspace*{-0.3cm} ~(1-\alpha)\log_2\left( 1 + \frac{|\qh_B^H\qw_B|^2}{\sigma_m^2+ p_m
    |h_m|^2}\right)\nonumber\\
\hspace*{-0.3cm}     \mbox{s.t.}~~  \hspace*{-0.3cm}& & \hspace*{-0.3cm}   (1-\alpha)\log_2\biggl( 1 + \frac{ p_m}{\sigma_b^2} \|\qh_M\|^2  - \biggr. \nonumber \\ & &\hspace*{1.5cm}  \biggl.  \frac{\frac{p_m}{\sigma_b^2} |\qh_M^H\qH_B\qw_B |^2  }{\sigma_b^2+ \|\qH_B\qw_B\|^2}  \biggr) =R_B,\nonumber\\
          \hspace*{-0.3cm}& & \hspace*{-0.3cm} || \qw_B||^2 \leq P.
         \eea
 Since $\log(1+x)$ is a monotonically increasing function of $x$ and the denominator of $x\triangleq  \frac{|\qh_B^H\qw_B|^2}{\sigma_m^2+ p_m
    |h_m|^2}$  is independent of $\qw_B$, (\ref{eq:ProbForm2}) can be solved via
 \bea
  \label{eq:ProbForm3}
    \max_{ \qw_B }~ \hspace*{-0.3cm}& \hspace*{-0.3cm}&   |\qh_B^H\qw_B|^2  \nonumber\\
    \mbox{s.t.}~~ \hspace*{-0.3cm}& \hspace*{-0.3cm}& \frac{|\qh_M^H\qH_B\qw_B |^2  }{\sigma_b^2+ \|\qH_B\qw_B\|^2}= \Gamma_B,\\
    \hspace*{-0.3cm}& \hspace*{-0.3cm}& \|\qw_B\|^2 \le P\nonumber
   \eea
where  $\Gamma_B\triangleq || \qh_M||^2-\frac{\sigma_b^2}{p_m}\left[ 2^{\frac{R_B}{1-\alpha}} -1 \right]$. It is clear that the objective function in \eqref{eq:ProbForm3} is maximized when $||\qw_B||^2=P$. This optimization problem is non-convex because it is the maximization of a quadratic function with a quadratic equality constraint. Moreover, to the best of our knowledge, \eqref{eq:ProbForm3}  does not admit a closed-form solution. However, it can be efficiently and optimally solved using semi-definite programming.  For this purpose, we define $\qV_B\triangleq \qw_B\qw_B^H$ and relax the rank-one constraint, ${\rm rank}(\qV_B)=1$. The relaxed form of  (\ref{eq:ProbForm3})  is given by 
 \bea
   \label{eq:ProbForm4}
    \max_{\qV_B}~ \hspace*{-0.3cm}& & \hspace*{-0.3cm} f(\alpha)=  {\rm tr}(\qV_B \qh_B \qh_B^H )\nonumber\\
    \mbox{s.t.}~ \hspace*{-0.3cm}& & \hspace*{-0.3cm} {\rm tr}(\qV_B \qH_B^H \qh_M\qh_M^H\qH_B )=\Gamma_B\left( \sigma_b^2+  {\rm tr}(\qV_B \qH_B^H \qH_B )\right),\nonumber\\
   \hspace*{-0.3cm}& & \hspace*{-0.3cm} {\rm tr}(\qV_B ) =P , ~\qV_B\succeq 0.
   \eea
The optimization problem in (\ref{eq:ProbForm4}) is a standard SDR problem \cite{Boyd}. In the following, we show that its optimum solution is rank-one.
\hspace*{-0.3cm}\begin{proposition}
\label{ProofPropositionRankOne}
The rank-one optimum solution is  guaranteed in \eqref{eq:ProbForm4}.
\end{proposition}
\begin{proof}
The proof is based on Karush-Kuhn-Tucker (KKT) conditions and given in Appendix B.
\end{proof}
Let $\qV_B^*$ be the optimum solution of \eqref{eq:ProbForm4}. Since  $\qV_B^*$ is a rank-one matrix,  the optimum solution $\qw_B^*$  is obtained as $\qw_B^*=\sqrt{P}{\tilde \qu}{\tilde \qu}^H$, where ${\tilde \qu}$ is the eigenvector corresponding to the non-zero eigenvalue of $\qV_B^*$ . 

\subsubsection{ Optimization of $\qw_B$ and $\alpha$}
In order to jointly optimize $\qw_B$ and $\alpha$, we solve the SDR problem in (\ref{eq:ProbForm4}) by using one-dimensional (or line search) search over $\alpha$. This line search can be confined to a small region of $\alpha$, and therefore, the number of required SDR optimizations can be significantly minimized. To illustrate this, let the objective function in (\ref{eq:ProbForm2}), for a given $\qw_B^{*}$, be defined as
\bea
\label{eq:ProbForm4Ex1}
f(\alpha)=(1-\alpha)\log_2\left( 1 +  \frac{{\tilde \beta}}{c+\frac{\alpha b}{1-\alpha}}\right)
\eea
where
\bea
  \label{eq:ProbForm4Ex2}
   {\tilde \beta}=  \frac{|\qh_B^H\qw_B^*|^2}{|h_m|^2}, c=\frac{\sigma_m^2}{|h_m|^2}, b=\eta P {\bar \lambda}\left(\qH_{BM}\qH_{BM}^H\right).
         \eea
The derivative of $f(\alpha)$ w.r.t. $\alpha$ can be written as 
\bea
 \label{eq:ProbForm4Ex3}
\frac{d f(\alpha)}{d \alpha}=-\log_2(g(\alpha))-\frac{b {\tilde \beta}  g(\alpha)^{-1}}{(1-\alpha)\log(2)}\left(c+\frac{\alpha b}{1-\alpha}\right)^{-2}
\eea
where $g(\alpha)= 1 +  \frac{ {\tilde \beta}}{c+\frac{\alpha b}{1-\alpha}}\geq 0, \forall \alpha \in [0, 1]$. It is clear from  (\ref{eq:ProbForm4Ex3}) that $\frac{d f(\alpha)}{d \alpha}<0$ for all $\alpha$, i.e., $f(\alpha)$ is a monotonically decreasing function of $\alpha$. This means that maximum value of the objective function in (\ref{eq:ProbForm4Ex1})  is achieved when $\alpha$ is minimum, provided that the equality constraint in (\ref{eq:ProbForm4}) is fulfilled. However, as $\alpha\rightarrow 0$, $\Gamma_B\rightarrow -\infty$, i.e., the chance that the SDR optimization problem (\ref{eq:ProbForm4}) is infeasible increases. Consequently, the optimum $\alpha$ is the minimum one for which  (\ref{eq:ProbForm4}) is feasible. The output $\qV_B$ of such feasible SDR provides the optimum ${\bf w}_B$.
 In a nutshell, the steps of the proposed algorithm ({\bf Algorithm 1}) for the joint optimization  problem can be summarized as follows:
\begin{itemize}
\item[ ] 1) Define a fine grid of $\alpha$ in steps of $\partial\alpha$. Start with $\alpha=0$.  
\item[ ] 2) Solve  (\ref{eq:ProbForm4})  with the increment of $\partial\alpha$.
\item[ ] 3) If feasible, stop and output $\alpha$ and $\qV_B$. 
\item[ ] 4) If not,  go to step (2).
\end{itemize}
Since this algorithm terminates as soon as (\ref{eq:ProbForm4}) is feasible, the search region of $\alpha$ is usually limited to a range of its small values. We use CVX toolbox \cite{Boyd} to solve  the SDR problem  (\ref{eq:ProbForm4}). For a given solution accuracy of $\epsilon>0$, the worst-case complexity of this optimization problem  is given by ${\mathcal O}\left(N_t^{4.5} \log\left(\frac{1}{\epsilon}\right)\right)$ \cite{LuoMa}. While executing {\bf Algorithm 1}, the SDR problem is solved for that particular $\alpha$ for which the problem is feasible. Note that,  by considering that $\Gamma_B$ should be positive and $||{\bf w}_B ||^2\leq P$, necessary conditions for the feasibility \cite{FengZhu} of (\ref{eq:ProbForm4}) can be checked before calling the CVX routine. Therefore, the worst-case computational complexity of executing {\bf Algorithm 1} is only ${\mathcal O}\left(N_t^{4.5} \log\left(\frac{1}{\epsilon}\right)\right)$. 

\subsection{Suboptimal Method}
As a suboptimal method of optimizing $\qw_B$ and $\alpha$, we consider the ZF approach \cite{Suraweera-TWC14}. This requires that 
   \be
    \label{eq:ProbForm5}
        \qw_B^H\qH_B^H \qh_M=0.
   \ee
\subsubsection{Optimization of $\qw_B$}
Substituting  \eqref{eq:ProbForm5} into   \eqref{eq:ProbForm1}, the resulting optimization problem is expressed as
\bea
 \label{eq:ProbForm6}
    \max_{ \{\qw_B,  \alpha\}}~  \hspace*{-0.3cm}& & \hspace*{-0.3cm} (1-\alpha)\log_2\left( 1 + \frac{|\qh_B^H\qw_B|^2}{\sigma_m^2+ p_m
    |h_m|^2}\right)\nonumber\\ 
    \mbox{s.t.}~  \hspace*{-0.3cm}& & \hspace*{-0.3cm} (1-\alpha)\log_2\left( 1 +  \frac{p_m}{\sigma_b^2} \|\qh_M\|^2     \right)= R_B,\nonumber\\
          \hspace*{-0.3cm}& & \hspace*{-0.3cm} p_m =  \frac{ \alpha\eta P{\bar \lambda}{ \left(\qH_{BM}\qH_{BM}^H\right)} }{1-\alpha}\\
           \hspace*{-0.3cm}& & \hspace*{-0.3cm}  \|\qw_B\|^2  \le P,~~ 0\le \alpha\le 1,\nonumber\\
 \hspace*{-0.3cm}& & \hspace*{-0.3cm}    \qw_B^H\qH_B^H \qh_M=0.\nonumber
         \eea
For a given $\alpha$,  the optimization of $\qw_B$ becomes
 \bea
  \label{eq:ProbForm7}
    \max_{\qw_B } &&  |\qh_B^H\qw_B|^2\nonumber\\
    \mbox{s.t.} &&   \|\qw_B\|^2   \le P \\
&&    \qw_B^H\qH_B^H \qh_M=0.\nonumber
         \eea
Using the standard Lagrangian multiplier method and after some manipulations,  a closed-form solution of $\qw_B$ is obtained as
         \be
          \label{eq:ProbForm8}
            \qw_B = \sqrt{P} \frac{\qB\qh_B}{\|\qB\qh_B\|}, \qB =
            \qI - \frac{ \qH_B^H \qh_M\qh_M^H\qH_B}{\| \qH_B^H \qh_M\|^2}
         \ee
        where $\qw_B$ does not depend on $\alpha$ and $\qB$ is a projection matrix, i.e., $\qB^H\qB=\qB$. Consequently, the corresponding objective function in (\ref{eq:ProbForm7})  is
 \be
  \label{eq:ProbForm9}
    |\qh_B^H\qw_B|^2 = P
    \frac{|\qh_B^H\qB\qh_B|^2}{\|\qB\qh_B\|^2} =P \qh_B^H\qB\qh_B.
 \ee

\subsubsection{Optimization of $\alpha$}
Denote the suboptimal beamformer solution of   (\ref{eq:ProbForm8})  by ${\bar \qw}_B^{*}$.  The remaining optimization problem is expressed as 
 \bea
  \label{eq:ProbForm10}
    \max_{  0\le \alpha \le 1 } &&  f(\alpha)\triangleq (1-\alpha)\log_2\left( 1 +  \frac{{\bar \beta}}{c+\frac{\alpha b}{1-\alpha}}\right)\nonumber\\
    \mbox{s.t.} && (1-\alpha)\log_2\left( 1 + \frac{\alpha}{1-\alpha}b\gamma\right)=R_B,
         \eea
      where  ${\bar{\beta}}=\frac{|\qh_B^H{\bar {\bf w}}_B^{*}|^2}{|h_m|^2}$ and  $\gamma=\frac{|| \qh_M||^2}{\sigma_b^2}$.  
 Note  that the optimum $\alpha$ would be zero if there is no equality constraint (or if the constraint is, $R_B=0$). This is because the objective function in (\ref{eq:ProbForm10}) is a monotonically decreasing function of $\alpha$.  In the presence of equality constraint with $R_B>0$, it is clear that the optimum $\alpha$ is the smallest value that satisfies the equality constraint.  
\hspace*{-0.3cm}\begin{proposition}
\label{ProofPropositionOptalpha}
When equality constraint is feasible (i.e., $R_B\le R_B^{\rm max}$) , the optimum $\alpha$ is given by
\bea
\label{eq:ProbForm13}
{\bar \alpha}^{\rm opt}\hspace*{-0.3cm}&=\hspace*{-0.3cm}&\frac{ - \displaystyle{\frac{1}{{\bar R}_B}} W\left( -\frac{ {\bar R}_B}{b\gamma} {\rm e}^{  {\bar  R}_B \left(1-\frac{1}{b\gamma}\right)}\right)-\frac{1}{b\gamma}}{1-\displaystyle{\frac{1}{{\bar R}_B}} W\left( -\frac{ {\bar R}_B}{b\gamma} {\rm e}^{   {\bar R}_B \left(1-\frac{1}{b\gamma}\right)}\right)-\frac{1}{b\gamma} }
\eea
where ${\bar R}_B=R_B\log(2)$ and $W(y)$  is the Lambert function \cite{Knuth}. 
\end{proposition}
\begin{proof}
The proof is  given in Appendix C.
\end{proof}
Since we use the performance of the HD mode as a benchmark, we end this section by making a remark on this mode.  In the HD mode, the information transmission period of $(1 - \alpha)$ is equally divided for the BS to MS and then the MS to BS communications. Moreover,  both BS and MS can employ all of their antennas for transmit and receive beamforming, as in the case of standard MIMO communication. However, in this case, the HD mode requires twice the RF chains required by the proposed FD approach which, in fact, is based on the {\it antenna conserved} (AC) condition \cite{Barghi}. Since RF chains are more expensive than the antennas, a comparison between this type of HD mode, which we refer to as {\it HD-AC} and the proposed FD methods is not fair. As such, we also consider another type of HD mode which uses the same number of the transmit and receive antennas (at each node) as in the FD mode (see Fig. 1-b), leading to the same number of RF chains. This type of HD mode is referred to as HD with {\it radio-frequency (RF) chain conserved} (RFC) condition, i.e., {\it HD-RFC}. Therefore, for the HD-AC approach, the BS and MS rate, respectively, are given by
\begin{eqnarray}
\label{eqn:HD1}
r_{B,H}&=&\frac{1-\alpha}{2}\log_2\left( 1+ \frac{\alpha}{1-\alpha}\frac{\eta P {\bar \lambda}^2(\qH_{BM}\qH^H_{BM})}{\sigma^2_b}\right)\nonumber\\
r_{M,H}&=&\frac{1-\alpha}{2}\log_2\left( 1+\frac{P {\bar \lambda}(\qH_{BM}\qH^H_{BM})}{\sigma^2_m}\right).
\end{eqnarray}
On the other hand, the respective BS and MS rates, under HD-RFC approach, are given by
 \begin{eqnarray}
\label{eqn:HD1N}
{\tilde r}_{B,H}&=&\frac{1-\alpha}{2}\log_2\left( 1+ \frac{\alpha}{1-\alpha}\frac{\eta P ||\qh_M||^2}{\sigma^2_b}\right)\nonumber\\
{\tilde r}_{M,H}&=&\frac{1-\alpha}{2}\log_2\left( 1+\frac{P||\qh_B||^2}{\sigma^2_m}\right).
\end{eqnarray}

\section{Joint Optimization with Partial CSI}
\label{sec:Section3}
In the previous section, the optimum beamformer and TS parameter are obtained by assuming that the instantaneous CSI is perfectly known. In particular, the assumption of having perfect instantaneous transmit CSI is idealistic due to the fact that each terminal,  in general, has to rely on the CSI fed back by the other terminal. In order to minimize the cost of CSI feedback, it is often preferred to pursue system design that requires only the knowledge of second-order statistics of the transmit CSI. Notice that, if the channel varies rapidly, this approach becomes somehow inevitable, since the optimal parameters designed on the basis of previously acquired CSI becomes outdated quickly \cite{SZhou}, \cite{ChaliseSuraweera2013}. With these motivations, we consider that each terminal  knows its own LI channel and receive CSI,  but  only the second-order statistics (more specifically channel covariance matrix) of the transmit CSI. Although degradation in system performance is inevitable due to partial CSI, such degradation cannot be analytically quantified. However, numerical results (not included due to space constraints) show that the performance degradation can be significant, and thus, an improved design approach is necessary to achieve a desired level of performance.

 In the first phase, the BS transmits an energy signal isotropically, and 
 the harvested energy at the MS is
 \be
    E = \frac{\alpha\eta P{\rm tr}\left(\qH_{BM}\qH_{BM}^H\right)}{N_t}.
     \ee
From the received signal  (\ref{eq:sigmod1}),  the  achievable ergodic rate at the MS  is given by
\bea
  r_M  \hspace*{-0.3cm} &= \hspace*{-0.3cm}& (1-\alpha)E_{\qh_B} \left[\log_2\left( 1 + \frac{|\qh_B^H\qw_B|^2}{\sigma_m^2+ p_m |h_m|^2}\right)\right],
\eea
whereas the outage probability at the BS,  defined as $P_{out, B}={\rm Pr}\left\{r_B \leq \gamma_B  \right\}$, is expressed as
 \bea
    P_{out, B} \hspace*{-0.3cm}&= \hspace*{-0.3cm}& {\rm Pr}\biggl\{ (1-\alpha)\log_2\biggl( 1 +  p_m \qh_M^H \biggr. \biggr. \nonumber\\
\biggl. \biggl.
& & (\sigma_b^2\qI+  \qH_B\qw_B\qw_B^H\qH_B^H )^{-1}\qh_M\biggr) \leq \gamma_B \biggr\},
  \eea
where $\gamma_B$ is a predefined threshold value for  BS information rate.  
\subsection{Problem Formulation}
  The objective is to maximize the ergodic rate of the MS, while confirming that the outage probability at the BS does not exceed a certain value, $\rho$,  where $0 \leq \rho\leq 1$.  
  This is achieved by solving the following optimization problem
 \bea
 \label{eq:Rob1}
    \max_{ \{\qw_B, p_m, \alpha\} }~  \hspace*{-0.3cm}& & \hspace*{-0.3cm} (1-\alpha)E_{\qh_B} \left[\log_2\left( 1 + \frac{|\qh_B^H\qw_B|^2}{\sigma_m^2+ p_m |h_m|^2}\right)\right]\nonumber\\
    \mbox{s.t.}~  \hspace*{-0.3cm}& & \hspace*{-0.3cm}  {\rm Pr}\biggl\{ (1-\alpha)\log_2\biggl( 1 +  p_m \qh_M^H \biggr. \biggr. \nonumber\\
\biggl. \biggl.
 \hspace*{-0.3cm}& & \hspace*{-0.3cm} (\sigma_b^2\qI+  \qH_B\qw_B\qw_B^H\qH_B^H )^{-1}\qh_M\biggr) \leq \gamma_B \biggr\} \leq \rho \nonumber\\
         \hspace*{-0.3cm}& & \hspace*{-0.3cm} p_m = E_{\qH_{BM}}\left[\frac{  \frac{\alpha\eta P{\rm tr}\left(\qH_{BM}\qH_{BM}^H\right)}{N_t} 
         }{1-\alpha}\right],\\
  \hspace*{-0.3cm}& & \hspace*{-0.3cm} ||\qw_B||^2\leq P, ~0\le \alpha \le 1, \nonumber
         \eea
where the equality constraint on $p_m$ can be further expressed as
\bea
p_m= \frac{  \frac{\alpha\eta P{\rm tr} \left(E_{\qH_{BM}}\left[\qH_{BM}\qH_{BM}^H\right]\right)}{N_t} 
         }{1-\alpha}.
\eea
 In order to solve the optimization problem in (\ref{eq:Rob1}), the ergodic  MS rate and the BS outage probability need to be
 first derived. 

\subsection{Ergodic Rate and Outage Probability}
Let $r_M=\frac{{\tilde r}_M(1-\alpha)}{\log(2)}$, where
\bea
{\tilde r}_M&=& E_{\qh_B} \left[\log\left( 1 + \frac{|\qh_B^H\qw_B|^2}{\sigma_m^2+ p_m |h_m|^2}\right)\right].
\eea
Assuming Rayleigh fading, next we derive the exact closed-form expressions for ${\tilde r}_M$ and $P_{out, B}$. \\

Let $\qh_B$ be expressed as  $\qh_B=\qR_B^{\frac{1}{2}}\qh_ {B,w}$,  where the elements of $\qh_ {B,w}$ are ZMCSCG with unit-variance,
$\qR_ B$ is the covariance matrix of $\qh_B$. Note that the effect of distance dependent attenuation is lumped into $\qR_ B$. Then,  the random variable (RV) $V \triangleq | \qh_B^H\qw_B|^2$ can be written as  $V=\qh_ {B,w}^H \qR_ {B}^ {\frac {1} {2}} \qw_B \qw_B^H\qR_ {B}^ {\frac {1} {2}} \qh_{B,w}=\qh_ {B,w}^H{\bar \qU}{\bar {\boldsymbol \Lambda}}{\bar \qU}^H \qh_ {B,w}$, where  ${\bar \qU}$ is the matrix of eigenvectors and ${\bar {\boldsymbol \Lambda}}$ is a diagonal matrix of eigenvalues of  the matrix $\qR_ {B}^ {\frac {1} {2}} \qw_B \qw_B^H\qR_ {B}^ {\frac {1} {2}}$. Since it is rank-one matrix, only one diagonal element of  ${\bar {\boldsymbol \Lambda}}$ is non-zero, which is  $\lambda\triangleq || \qR_ {B}^ {\frac {1} {2}} \qw_B \qw_B^H\qR_ {B}^ {\frac {1} {2}}||$. Hence, $V=\lambda  |  {\tilde h}_ {B,w,n}|^2$,  where ${\tilde h}_ {B,w,n}$ is the element of $ {\bar {\bf U}}^H\qh_ {B,w}$ corresponding to  the non-zero eigenvalue of $\qR_ {B}^ {\frac {1} {2}} \qw_B \qw_B^H\qR_ {B}^ {\frac {1} {2}}$ (i.e., $\lambda$). Since  $|  {\tilde h}_ {B,w,n}|^2$ is an exponentially distributed RV with unit parameter,
the probability density function (PDF) of $V$ is given by
\bea
f_V(v)=\frac{1}{\lambda}{\rm e}^{-\frac{v}{\lambda}}.
 \eea 
Let $c_f=\frac{1}{\sigma_m^2+p_m|h_m|^2}$ and $ {\bar V}=c_fV$. Then, the PDF of ${\bar V}$ is given by
\bea
\label{eq:pdf1}
f_{\bar V}( {\bar v} )=\frac{1} {\lambda c_f}  {\rm e}^ {-\frac{\bar v}{\lambda c_f} }.  
 \eea 
Using the PDF of (\ref{eq:pdf1}), we get ${\tilde r}_M\triangleq {\rm E}\left\{\log(1+{\bar V})\right\}$ as 
\bea
\hspace*{-0.5cm}{\tilde r}_M\hspace*{-0.3cm}&=\hspace*{-0.3cm}&\frac{1}{\lambda c_f} \int_{0}^{\infty}\log\left(1+ {\bar v}\right) {\rm e}^ {-\frac{\bar v}{\lambda c_f} } d{\bar v}\nonumber\\
\hspace*{-0.3cm}&=\hspace*{-0.3cm}& {\rm e}^{ \frac{\sigma_m^2+p_m|h_m|^2}{{\bf w}_B^H\qR_B\qw_B}}{\rm E}_1\left(   \frac{\sigma_m^2+p_m|h_m|^2}{{\bf w}_B^H\qR_B\qw_B}\right),
\eea
where we use \cite[eqns. 4.331.2, 8.211.1]{Ryzhik} and $E_1(\cdot)$ is the exponential integral \cite[p. 228]{Stegun}.\\
\hspace*{-0.3cm}\begin{proposition}
\label{ProofPropositionExact_RateB}
A closed-form expression for $P_{out, B}$ is given by
\bea
P_{out, B}=\sum_{i=1}^{L} a_i\lambda_i\left[ 1- {\rm e}^{-\frac{{\bar \gamma}}{\lambda_i}}\right],
\eea
where
\bea
{\bar \gamma} & = & \frac{1}{p_m} \left[ 2^{\frac{\gamma_B}{1-\alpha}}-1 \right],\nonumber\\
a_i &= &\frac{\lambda_i^{L-2}}{\prod_{j=1, j\neq i}^{L} (\lambda_i-\lambda_j)},
\eea
and  $\left\{\lambda_i\right\}_{i=1}^{L}$ are the $L$ distinct eigenvalues of the matrix
\bea
{\boldsymbol \Phi}=\qR_M^{\frac{1}{2}}\left(\sigma_b^2\qI+\qH_B\qw_B\qw_B^H\qH_B^H\right)^{-1}\qR_M^{\frac{1}{2}},
\eea
with $\qR_M$  being the covariance matrix of $\qh_M$,  i.e., $\qh_M=\qR_M^{\frac{1}{2}}\qh_{M,w}$, and the elements of $\qh_{M,w}$ are ZMCSCG. 
\end{proposition}
\begin{proof}
The proof is  given in Appendix D.
\end{proof}

\subsection{Optimization}
 With the derived expressions for $r_M$ and $P_{out, B}$, the objective is to maximize $r_M$ while keeping $P_{out, B}$ less than a certain value $\rho$. This is mathematically expressed as
\bea
\label{FCSI_eq1New}
  \max\limits_{\{{\bf w}_B, \alpha\}}  \hspace*{-0.3cm}& & \hspace*{-0.3cm} \frac{1-\alpha}{\log(2)} {\rm e}^{ \frac{\sigma_m^2+p_m|h_m|^2}{{\bf w}_B^H\qR_B\qw_B}}{\rm E}_1\left(   \frac{\sigma_m^2+p_m|h_m|^2}{{\bf w}_B^H\qR_B\qw_B}  \right) \nonumber\\
 \mbox{s.t.}~ \hspace*{-0.3cm}& & \hspace*{-0.3cm} \sum_{i=1}^{L}a_i\lambda_i\left[ 1- {\rm e}^{-\frac{{\bar \gamma}}{\lambda_i}}\right]  \leq \rho, \nonumber\\
 \hspace*{-0.3cm}& & \hspace*{-0.3cm} ||\qw_B||^2\leq P,~~0\le\alpha\le 1,
\eea 
which is a very difficult optimization problem since $\{\lambda_i\}$ are eigenvalues of $\qR_M^{\frac{1}{2}}\left(\sigma_b^2\qI+\qH_B\qw_B\qw_B^H\qH_B^H\right)^{-1}\qR_M^{\frac{1}{2}}$ and $\{a_i\}$ are complicated functions of $\{\lambda_i\}$. 
In order to solve the problem (\ref{FCSI_eq1New}), we first prove the following lemma. 
\begin{lemma}
Let $f(x)={\rm e}^{x}{\rm E}_1(x)$ with $x\geq 0$. Then, $f(x)$ is a monotonically decreasing function of $x$. 
\end{lemma}
\begin{proof}
The first-order derivative of $f(x)$ is given by
\bea
\label{FCSI_eq3}
\frac{d f(x)}{d x}={\rm e}^{x}{\rm E}_1(x)-{\rm e}^{x}{\rm E}_0(x)={\rm e}^{x}{\rm E}_1(x)-\frac{1}{x},
\eea
where we use that \cite[p. 230]{Stegun}
\bea
\label{FCSI_eq4}
\frac{d {\rm E}_1(x)}{d x}=-{\rm E}_0(x)=-\frac{{\rm e}^{-x}}{x}.
\eea
On the other hand, ${\rm E}_1(x)$ can be lower bounded as  \cite[p. 229]{Stegun}
\bea
\label{FCSI_eq5}
{\rm E}_1(x)\le {\rm e}^{-x}\log\left( 1+\frac{1}{x}\right), x \ge 0.  
\eea
Applying (\ref{FCSI_eq5}) in (\ref{FCSI_eq3}) leads to
\bea
\label{FCSI_eq6}
\frac{d f(x)}{d x} \le \log\left( 1+\frac{1}{x}\right)-\frac{1}{x} \le 0,
\eea
since $\log\left( 1+\frac{1}{x}\right) \le \frac{1}{x}$. Therefore, ${\rm e}^{x}{\rm E}_1(x)$ is a monotonically decreasing function of $x$. 
\end{proof}
Applying  Lemma 1 to (\ref{FCSI_eq1New}), it is evident that the objective function  monotonically decreases with $\frac{\sigma_m^2+p_m|h_m|^2}{\qw_B^H\qR_B\qw}$. For a given $\alpha$ and known $|h_m|^2$, maximizing the objective function is equivalent to minimizing  $\frac{1}{\qw_B^H\qR_B\qw}$. Consequently, for a given $\alpha$, the optimization problem  (\ref{FCSI_eq1New}) can be expressed as
\bea
\label{FCSI_eq2New}
 \min\limits_{{\bf w}_B}~\hspace*{-0.3cm}& \hspace*{-0.3cm}& \frac{1}{{\bf w}_B^H\qR_B\qw_B} \nonumber\\
\mbox{s.t.}~\hspace*{-0.3cm}& \hspace*{-0.3cm}&\sum_{i=1}^{L}a_i\lambda_i \left[ 1- {\rm e}^{-\frac{{\bar \gamma}}{\lambda_i}}\right]  \leq \rho, \nonumber\\
\hspace*{-0.3cm}& \hspace*{-0.3cm}& ||\qw_B||^2\leq P.
\eea 
However, the optimization problem in (\ref{FCSI_eq2New}) is still not tractable,  due to the complicated constraint on the outage probability. As such, we derive an upper bound for $P_{out, B}$. This upper bound can be tightened by optimizing one of the parameters that will be clear in the sequel.  Notice that,  from  Appendix D,  $P_{out,B}$ is expressed as
\bea
\label{FCSI_eq7A}
P_{out, B}={\rm Pr}\left\{ \sum_{i=1}^{L} \lambda_i|{\tilde h}_{M, w, i}|^2 \leq {\bar \gamma}  \right\},
\eea
where  $|{\tilde h}_{M, w, i}|^2$ is an exponentially distributed RV with unit parameter. Applying Chernoff's bound \cite{Simon}, (\ref{FCSI_eq7A}) can be upper bounded as
\bea
\label{FCSI_eq7}
P_{out, B} & \leq &  {\rm E}_{{\tilde h}_{M, w}}\left[ {\rm e}^{-\beta \left(\sum_{i=1}^{L} \lambda_i |{\tilde h}_{M, w, i}|^2 -{\bar \gamma}\right)}\right]\nonumber\\
&=& {\rm e}^{\beta {\bar \gamma}} {\rm E}_{{\tilde h}_{M, w}}\left[ {\rm e}^{-\beta \sum_{i=1}^{L} \lambda_i |{\tilde h}_{M, w, i}|^2 }\right],
\eea
where $\beta\geq 0$ and ${\rm E}_{{\tilde h}_{M, w}}$ denotes mathematical expectation w.r.t. the random variables $\{{\tilde h}_{M, w, i}\}$. Since these variables are independent,
\bea
\label{FCSI_eq8}
P_{out, B} & \leq & {\rm e}^{\beta {\bar \gamma}}  \prod_{i=1}^{L} {\rm E}_{{\tilde h}_{M, w,i}}\left[ {\rm e}^{-\beta \lambda_i |{\tilde h}_{M, w, i}|^2 }\right]\nonumber\\
&=& {\rm e}^{\beta {\bar \gamma}}  \prod_{i=1}^{L} \frac{1}{1+\beta \lambda_i},
\eea
where in the second step we utilize the fact that  the PDF of $Y=|{\tilde h}_{M, w, i}|^2$ is given by $f_Y(y)=e^{-y}$. Since $\{\lambda_i\}$ are the eigenvalues of ${\boldsymbol \Phi}$, (\ref{FCSI_eq8}) is readily expressed as
\bea
\label{FCSI_eq9}
P_{out, B} & \leq &  \frac{{\rm e}^{\beta {\bar \gamma}} }{{\rm det}({\bf I}+\beta {\boldsymbol \Phi})}.
\eea
Let  ${\tilde P}_{out, B}=\frac{{\rm e}^{\beta {\bar \gamma}} }{{\rm det}({\bf I}+\beta {\boldsymbol \Phi})}$. The gap between $P_{out, B}$ and  ${\tilde P}_{out, B}$ can be minimized by computing $\min\limits_{\beta\geq 0} {\tilde P}_{out, B}$. With these results, the optimization problem (\ref{FCSI_eq2New}) for a given $\alpha$ is given by
\bea
\label{FCSI_eq3New}
  \min\limits_{{\bf w}_B}~ \hspace*{-0.3cm}& \hspace*{-0.3cm}& \frac{1}{{\bf w}_B^H\qR_B\qw_B} \nonumber\\
 \mbox{s.t.}~\hspace*{-0.3cm}& \hspace*{-0.3cm}& \min\limits_{\beta\geq 0} \frac{{\rm e}^{\beta {\bar \gamma}} }{{\rm det}({\bf I}+\beta \qR_M\left(\sigma_b^2\qI+\qH_B\qw_B\qw_B^H\qH_B^H\right)^{-1} )}  \leq \rho, \nonumber\\
\hspace*{-0.3cm}& \hspace*{-0.3cm}& ||\qw_B||^2\leq P,
\eea 
which can be equivalently expressed as
\bea
\label{FCSI_eq4New}
  \min\limits_{{\bf w}_B, \beta\geq 0}~\hspace*{-0.3cm}& \hspace*{-0.3cm}& \frac{1}{{\bf w}_B^H\qR_B\qw_B} \nonumber\\
\mbox{s.t.}~ \hspace*{-0.3cm}& \hspace*{-0.3cm}& \rho {\rm det}({\bf I}+\beta \qR_M\left(\sigma_b^2\qI+\qH_B\qw_B\qw_B^H\qH_B^H\right)^{-1} ) \geq {\rm e}^{\beta {\bar \gamma}}, \nonumber\\
\hspace*{-0.3cm}& \hspace*{-0.3cm}& ||\qw_B||^2\leq P.
\eea 
Using Sherman-Morrison formula \cite{Hager}, we obtain
\bea
\label{FCSI_eq5New}
& &\hspace*{-2cm}  \qR_M^{\frac{1}{2}}\left[ \sigma_b^2\qI+\qH_B\qw_B\qw_B^H\qH_B^H\right]^{-1}\qR_M^{\frac{1}{2}}=\nonumber\\
& &  \frac{\qR_M}{\sigma_b^2}-\frac{1}{\sigma_b^2}\frac{ {\tilde {\bf a}}{\tilde {\bf a}}^H }{\sigma_b^2+||\qH_B\qw_B||^2},
\eea 
where ${\tilde {\bf a}}=\qR_M^{\frac{1}{2}}\qH_B\qw_B$. Substituting (\ref{FCSI_eq5New}) into ${\tilde f}\triangleq \rho {\rm det}({\bf I}+\beta \qR_M^{\frac{1}{2}}\left(\sigma_b^2\qI+\qH_B\qw_B\qw_B^H\qH_B^H\right)^{-1}\qR_M^{\frac{1}{2}} )$, ${\tilde f}$ can be expressed as
\bea
\label{FCSI_eq6New}
{\tilde f}=\rho{\rm det}\left[ \frac{(\qI+\frac{\beta}{\sigma_b^2} \qR_M)(\sigma_b^2+||\qH_B\qw_B||^2)-\frac{\beta}{\sigma_b^2}  {\tilde {\bf a}}{\tilde {\bf a}}^H}{\sigma_b^2+||\qH_B\qw_B||^2 }\right],
\eea
which, due to the fact that ${\rm det}(c{\bf A})=c^n{\rm det}({\bf A})$ for any matrix ${\bf A}$ of size $n\times n$, is expressed as
\bea
\label{FCSI_eq7New}
{\tilde f}=\frac{\rho{\rm det}\left( (\qI+\frac{\beta}{\sigma_b^2}\qR_M)(\sigma_b^2+||\qH_B\qw_B||^2)-\frac{\beta}{\sigma_b^2}  {\tilde {\bf a}}{\tilde {\bf a}}^H   \right) }{(\sigma_b^2+||\qH_B\qw_B||^2)^{N_r}}.\nonumber
\eea 
Thus, the optimization problem in (\ref{FCSI_eq4New}) is expressed as
\bea
\label{FCSI_eq8New}
  \min\limits_{{\bf w}_B, \beta\geq 0} ~\hspace*{-0.3cm}& \hspace*{-0.3cm}& \frac{1}{{\bf w}_B^H\qR_B\qw_B} \nonumber\\
\mbox{s.t.}~\hspace*{-0.3cm}& \hspace*{-0.3cm}& \rho{\rm det}\biggl(  (\qI+\frac{\beta}{\sigma_b^2}\qR_M)(\sigma_b^2+||\qH_B\qw_B||^2)- \biggr. \biggr.\nonumber\\ 
\hspace*{-0.3cm}& \hspace*{-0.3cm}&   \frac{\beta}{\sigma_b^2}{\tilde {\bf a}}{\tilde {\bf a}}^H \biggr) \geq {\rm e}^{\beta {\bar \gamma}} (\sigma_b^2+||\qH_B\qw_B||^2)^{N_r}, \nonumber\\
\hspace*{-0.3cm}& \hspace*{-0.3cm}& ||\qw_B||^2\leq P, ~{\tilde {\bf a}}=\qR_M^{\frac{1}{2}}\qH_B\qw_B,
\eea 
which can be also written as
\bea
\label{FCSI_eq9New}
 \min\limits_{{\bf w}_B, \beta\geq 0}~\hspace*{-0.3cm}& \hspace*{-0.3cm}& \frac{1}{{\bf w}_B^H\qR_B\qw_B} \nonumber\\
\mbox{s.t.}~\hspace*{-0.3cm}& \hspace*{-0.3cm}& \biggl[ {\rm det}\biggl(  (\qI+\frac{\beta}{\sigma_b^2}\qR_M)(\sigma_b^2+||\qH_B\qw_B||^2)- \biggr.\biggr. \nonumber \\ 
\hspace*{-0.3cm}& \hspace*{-0.3cm}& \frac{\beta}{\sigma_b^2} {\tilde {\bf a}}{\tilde {\bf a}}^H \biggr)\biggr]^{\frac{1}{N_r}} \geq \rho^{-\frac{1}{N_r}} {\rm e}^{\frac{\beta {\bar \gamma}}{N_r}}  (\sigma_b^2+||\qH_B\qw_B||^2), \nonumber\\
\hspace*{-0.3cm}& \hspace*{-0.3cm}& ||\qw_B||^2\leq P, ~{\tilde {\bf a}}=\qR_M^{\frac{1}{2}}\qH_B\qw_B.
\eea 
The optimization problem (\ref{FCSI_eq9New}) is still non-convex, even for a given $\beta$. Introducing, $\qW_B=\qw_B\qw_B^H, \qW_B\succeq 0$ and relaxing the rank-one constraint of $\qW_B$, we obtain the following optimization problem
\bea
\label{FCSI_eq10New}
 \min\limits_{{\bf W}_B, \beta, \alpha} ~\hspace*{-0.3cm}& \hspace*{-0.3cm}& \frac{1}{ {\rm tr}( \qR_B\qW_B)}\nonumber\\
\mbox{s.t.}~\hspace*{-0.3cm}& \hspace*{-0.3cm}& \biggl[{\rm det}\biggl( (\qI+\frac{\beta}{\sigma_b^2}\qR_M)(\sigma_b^2+{\rm tr}(\qW_B\qH_B^H\qH_B))-\biggr. \biggr. \nonumber\\  & & \biggl. \biggl. \frac{\beta}{\sigma_b^2}{\tilde {\bf A}} \biggr)\biggr]^{\frac{1}{N_r}}  \geq \rho^{-\frac{1}{N_r}}  {\rm e}^{\frac{\beta {\bar \gamma}}{N_r}} (\sigma_b^2+{\rm tr}(\qW_B \qH_B^H\qH_B)) \nonumber\\
\hspace*{-0.3cm}& \hspace*{-0.3cm}& {\rm tr}(\qW_B)\leq P,~{\tilde {\bf A}}=\qR_M^{\frac{1}{2}}\qH_B\qW_B\qH_B^H\qR_M^{\frac{1}{2}},\\
\hspace*{-0.3cm}& \hspace*{-0.3cm}& {\bf W}_B\succeq 0, ~\beta\geq 0, ~0\le \alpha\le 1,\nonumber
\eea 
which is a convex optimization problem for a given $\alpha$ and $\beta$. This problem can be solved using the toolbox CVX \cite{Boyd}. In contrast to the perfect CSI case, the optimum solution of ${\bf W}_B$ in (\ref{FCSI_eq10New}) cannot be analytically guaranteed to be rank-one. If the optimum ${\bf W}_B$  is not rank-one, approximate rank-one solutions can be obtained using the randomization methods \cite{BatuLuc}. However, in all  simulation examples considered in Section \ref{sec:Section3}, we have not encountered a case in which the rank of the optimum  ${\bf W}_B$  is not rank-one{\footnote{Since we optimize $\beta$ to minimize the gap between exact outage probability and its upper bound, whenever  ${\bf W}_B$  is rank-one,  the solutions of  (\ref{FCSI_eq10New}) will also be the solutions of the original problem. This suggests that the proposed method gives close to optimum solutions. A more systematic way of its verification is an interesting work, but demands a significant level of a new task that is beyond the scope of this paper.}}. Note that  (\ref{FCSI_eq10New})  is not jointly convex w.r.t. $\alpha$, $\beta$, and $\qW_B$. A two-dimensional search over $\alpha$ and $\beta$ is required for solving this problem. However, we show that the required search space can be reduced. First note that 
\bea
\label{FCSI_eq11New}
\min\limits_{\beta>0} \frac{{\rm e}^{\beta {\bar \gamma}}}{{\rm det}(\qI+\beta {\boldsymbol \Phi})}=\min\limits_{\beta>0}{\rm e}^{\beta {\bar \gamma}-\log{\rm det}(\qI+\beta {\boldsymbol \Phi})},
\eea
which means that the exponent of the above function can be minimized. As such, the derivative of this exponent w.r.t. $\beta$ is
\bea
\label{FCSI_eq12New}
\frac{\partial [\beta {\bar \gamma}-\log{\rm det}(\qI+\beta {\boldsymbol \Phi})]}{\partial \beta}={\bar \gamma}-{\rm tr}( (\qI+\beta {\boldsymbol \Phi})^{-1}{\boldsymbol \Phi}).
\eea
After equating (\ref{FCSI_eq12New}) to zero, we get
\bea
\label{FCSI_eq13New}
{\bar \gamma}={\rm tr}( (\qI+\beta {\boldsymbol \Phi})^{-1}{\boldsymbol \Phi})=\sum_{i=1}^{L} \frac{\lambda_i}{1+\beta \lambda_i}.
\eea
Let $\beta^{*}$ be the solution of $\beta$ in (\ref{FCSI_eq13New}). Then, it is clear that ${\bar \gamma}\leq \frac{L}{\beta^*}$, i.e., $\beta^{*}\leq \frac{L}{\bar \gamma}$. Thus, the search space for $\beta$ can be confined to $[0, \frac{L}{ {\bar\gamma}}]$. On the other hand,  $\alpha$ takes only values between $0$ and $1$. Thus, the optimization problem (\ref{FCSI_eq10New}) can be efficiently solved. Since (\ref{FCSI_eq10New}) is an SDR problem for given $\alpha$ and $\beta$ and infeasibility conditions can be checked before calling the CVX routine, the worst-case computational complexity of this problem is approximately given $N_{\alpha}N_{\beta} {\mathcal O}\left(N_t^{4.5} \log\left(\frac{1}{\epsilon}\right)\right)$ where  $N_{\alpha}N_{\beta} $ denotes the number of points of the two-dimensional grid over $\alpha$ and $\beta$.

We end this subsection with the following remarks. In the case of MS with more than two antennas, in addition to the BS beamformer optimization, the joint MS receive and transmit beamformer optimization problem can be considered, which can be equivalently formulated in terms of the MS transmit beamformer. Then the optimization algorithms proposed in this paper can be applied with some minor modifications. In particular, sub-optimum solutions can be obtained by using alternating optimization method. More specifically, for a given $\alpha$, the BS transmit beamformer can be optimized while fixing the MS transmit beamformer, whereas the latter can be optimized by fixing the former.  The optimum $\alpha$ can then be obtained via one-dimensional search over $\alpha$.
 
\section{Numerical Results and Discussion}
\label{sec:Section4}
In this section, numerical results are presented for both full and partial CSI cases. More specifically, in the former case, the  MS-BS rate regions obtained from the optimum ({\bf Algorithm 1}) and sub-optimum methods ((\ref{eq:ProbForm8}) and (\ref{eq:ProbForm13})) are compared. In the partial CSI case, the ergodic MS rate versus the BS outage probability region is obtained by using the proposed method (\ref{FCSI_eq10New}). In both cases, the performance of the HD approach is also shown as a benchmark. In all simulation results, we take $\eta=0.5$, $N=6$, change the value of $N_t$, and set $P$ to $0$ dBm and $10$ dBm. The distance between the BS and MS is set to 10 meters, whereas the path loss exponent is taken as $3$. Note that, in a typical FD system, the digital cancellation scheme should be able to cancel at least 50 dB of LI power \cite{DineshKatti}. Considering this, we take $\sigma^2_{h_b}=\sigma^2_{h_m}=30$ dBm and $\sigma_b^2=\sigma_m^2=-70$ dBm, so that the LI at the BS has to be cancelled by $80$ dB when $P=10$ dBm and $70$ dB when $P=0$ dBm.
\subsection{Full CSI}
\vspace*{-0.1cm}
In this case, the channel coefficients for all channels are taken as ZMCSCG RVs. All results correspond to averaging of $100$ independent channel realizations. The BS rate is varied from $0$ to $R_B^{\rm max}$, where $R_B^{\rm max}$ is computed as in Appendix A. The HD-AC and HD-RFC schemes of the HD mode are compared with the FD mode. Fig. \ref{fig1:rate:region} shows the rate regions obtained with the optimum and sub-optimum methods for $N_t=4$ and $5$, when $P=0$ dBm, whereas the corresponding regions for $P=10$ dBm are shown in  Fig. \ref{fig2:rate:region}.  As a benchmark, the achieved BS-MS rate regions are also shown for the HD mode. 
\begin{figure}[t]
  \centering
	\includegraphics[width=0.9\columnwidth]{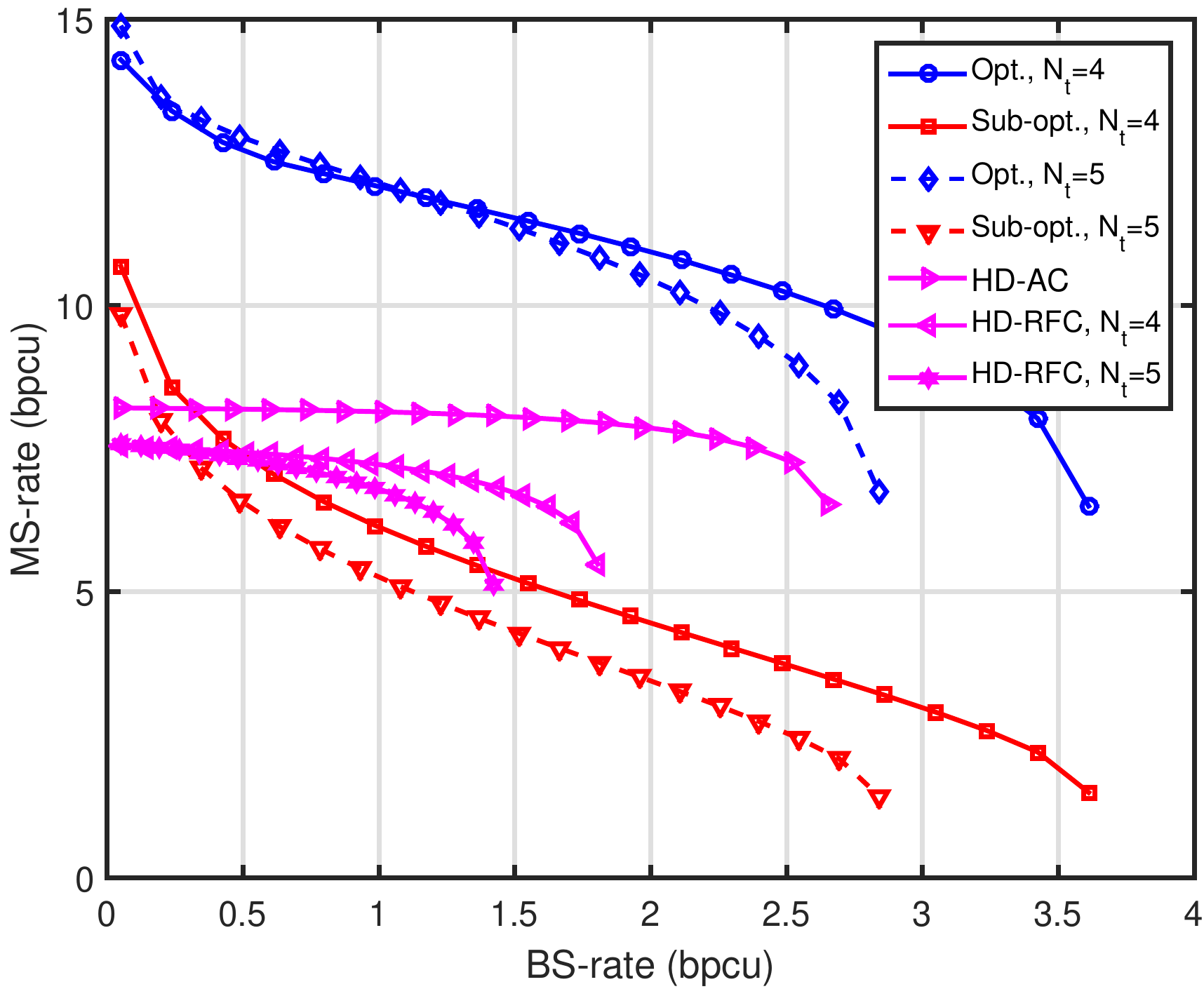}
\vspace*{-0.3cm}
 \caption{Comparison of rate regions with $P=0$ dBm, $\sigma^2_{h_b}=\sigma^2_{h_m}=30$ dBm and $N_t=4, 5$.}
\label{fig1:rate:region}
\end{figure}
It can be observed from Figs.~\ref{fig1:rate:region} and \ref{fig2:rate:region} that the maximum value of the MS rate (in bits per channel use (bpcu)) is obtained when $R_B$ is minimum, whereas the minimum value is obtained when $R_B$ takes maximum value. Moreover, as expected both the BS and MS rates increase when $P$ increases from $0$ dBm to $10$ dBm. Both figures show that the optimum method performs significantly better than the sub-optimum approach. In the proposed optimum method, when $N_t$ increases, the obtained maximum MS rate increases, whereas the obtained maximum BS rate decreases. This can be explained from the fact that increasing $N_t$ improves the transmit beamforming at the BS, which in turn is attributed for an increase in the MS rate. However, increase in $N_t$ decreases $N_r=N-N_t$ for a given $N_t$. This means that the LI rejection capability of the BS decreases which leads to a drop in the supported BS rate. All results also show that the proposed optimum method significantly outperforms HD modes that employ both AC and RFC approaches.  It is worthwhile to note that the boundaries of the MS-BS rate-regions remain relatively flat in the HD modes, although the corresponding maximum values of the BS and MS rates are smaller than those in the optimum and sub-optimum cases. When compared to the less expensive HD-RFC scheme, the sub-optimum method can be considered to provide a more flexible design. For example, in Fig. \ref{fig1:rate:region},  the HD-RFC, with $N_t=4$, gives the maximum BS-rate of about 1.8 bpcu. The corresponding MS-rate is about 5.5 bpcu which drops to zero beyond the BS-rate of 1.8 bpcu. However, the corresponding sub-optimum scheme supports the BS-rate up to 3.6 bpcu, although this is achieved with the MS-rate of only about 1.5 bpcu.
\begin{figure}[tb!]
  \centering
	 \includegraphics[width=0.9\columnwidth]{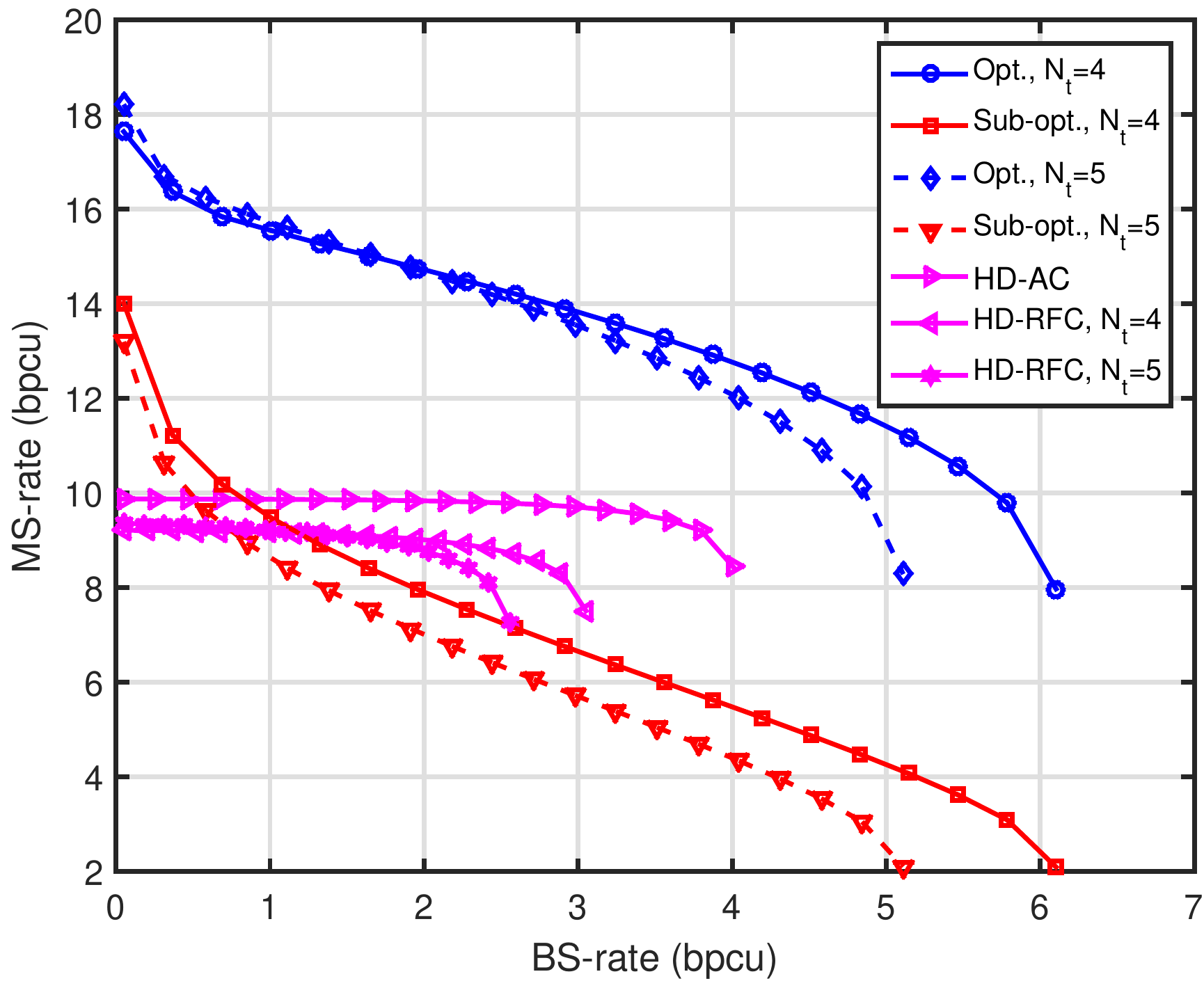}
\vspace*{-0.3cm}
 \caption{Comparison of rate regions with $P=10$ dBm,  $\sigma^2_{h_b}=\sigma^2_{h_m}=30$ dBm, and $N_t=4, 5$.}
\label{fig2:rate:region}
\end{figure}

The  rate regions of the optimum and sub-optimum methods with different values of $N_t$ are shown in Fig. \ref{fig3:rate:region} and  Fig. \ref{fig4:rate:region} for $P=0$ dBm and $P=10$ dBm, respectively. From these figures, similar observations can be made as in Figs. \ref{fig1:rate:region} and  \ref{fig2:rate:region}.
\begin{figure}[tb!]
  \centering
	\includegraphics[width=0.9\columnwidth]{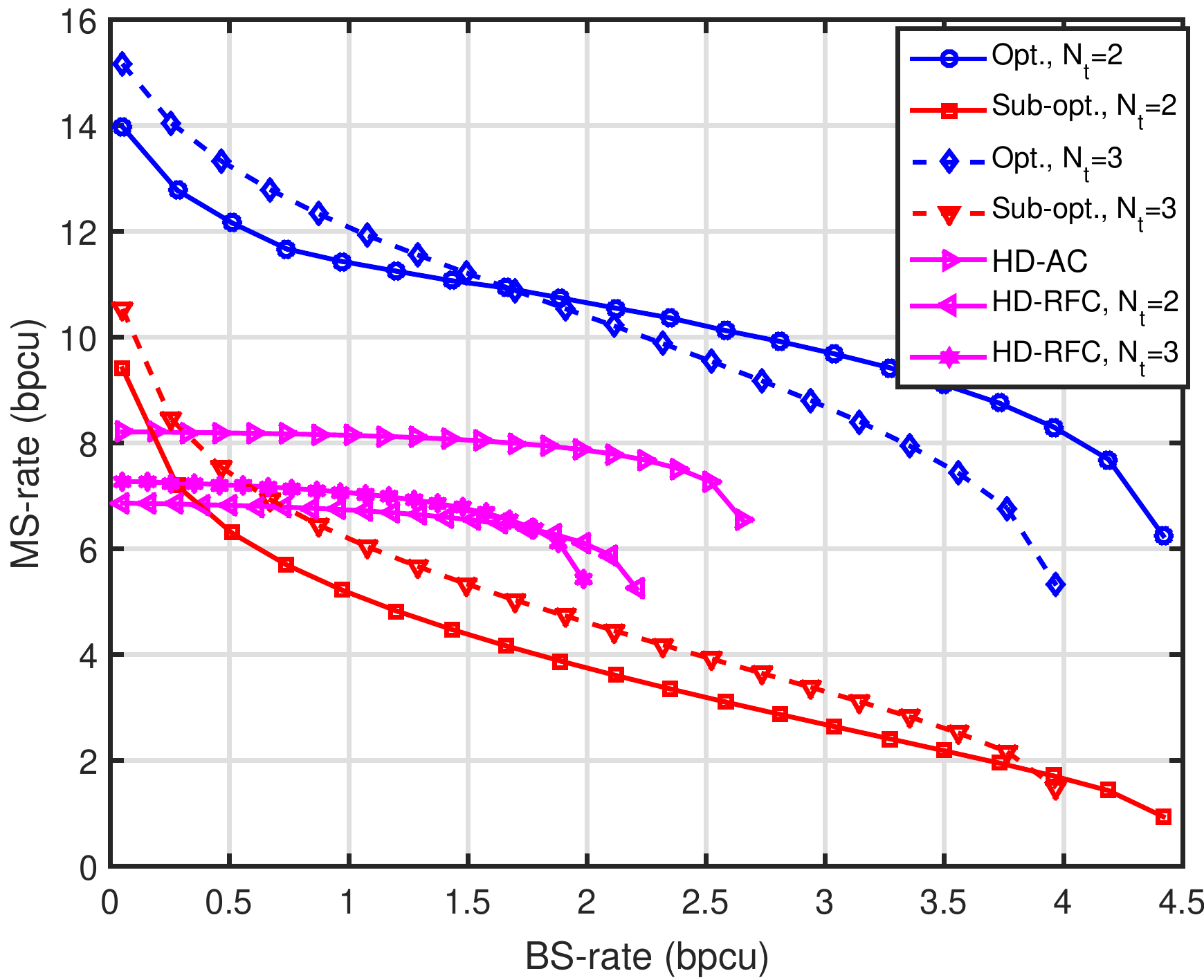}
	\vspace*{-0.3cm}
 \caption{Comparison of rate regions with $P=0$ dBm,  $\sigma^2_{h_b}=\sigma^2_{h_m}=30$ dBm, and  $N_t=2, 3$.}
\label{fig3:rate:region}
\end{figure}

\begin{figure}[tb!]
  \centering
	\includegraphics[width=0.9\columnwidth]{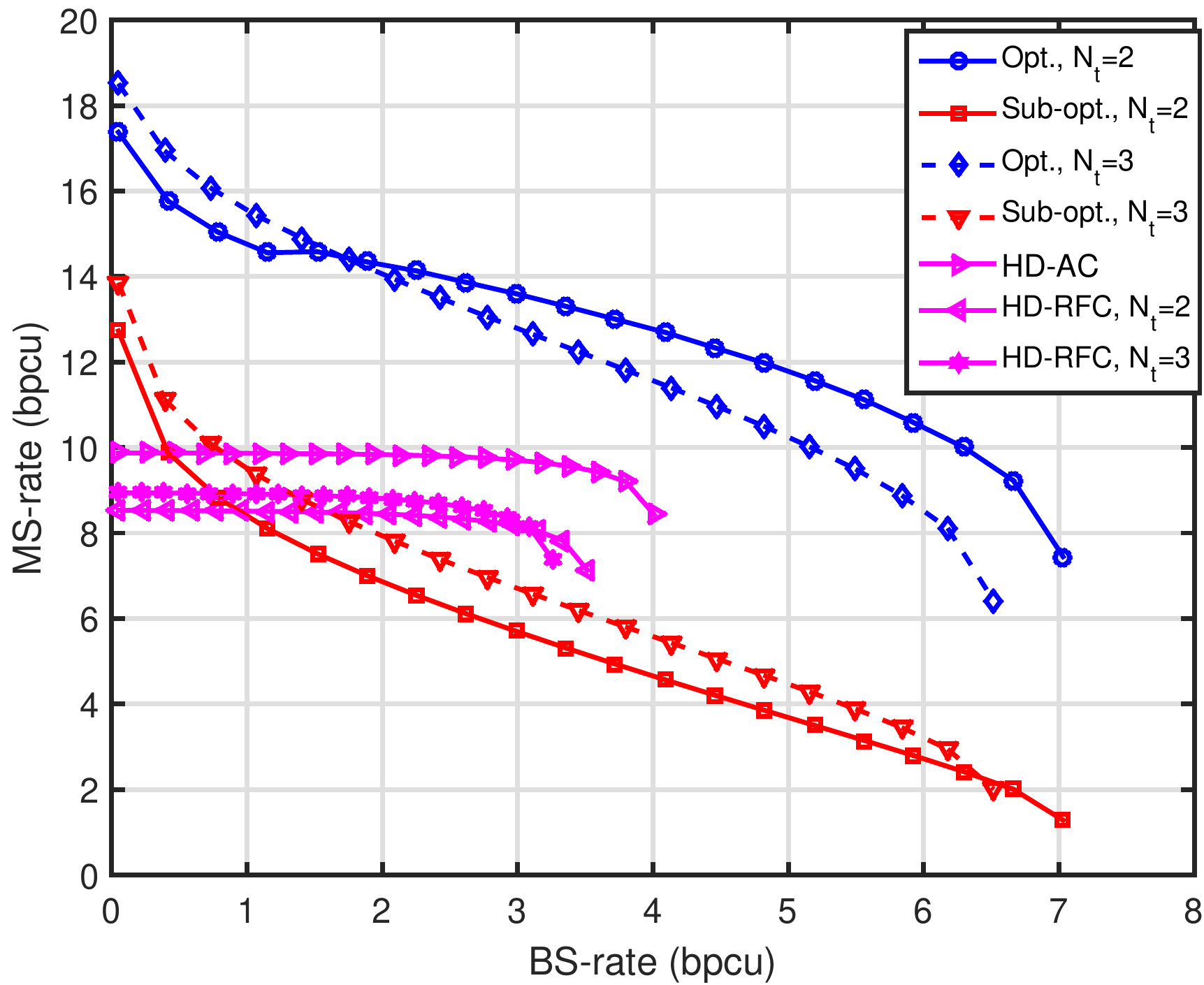}
\vspace*{-0.3cm}
 \caption{Comparison of rate regions $P=10$ dBm,  $\sigma^2_{h_b}=\sigma^2_{h_m}=30$ dBm, and  $N_t=2, 3$.}
\label{fig4:rate:region}
\end{figure}

\subsection{Partial CSI }
In this subsection, we first compare the analytical expressions of the BS outage probability and ergodic MS rate with the simulations. Considering that the BS and MS are often surrounded by multiple scatterers and the angle of arrival/departure undergo some spreading, the spatial covariance matrices $\qR_B$ and $\qR_M$ are modeled according to \cite{Mats} as 
\begin{eqnarray}
[\qR_{B/M}]_{m,n}=\frac{ {\rm e}^{{\rm j}\pi(m-n) \sin\theta_{B/M}}} {d^\tau} {\rm e}^{-\left( \pi(m-n)\sigma_{\theta}^{B/M}\cos\theta_{B/M} \right)^2/2}, \nonumber
\end{eqnarray}
where $[{\bf X} ]_{m,n}$ denotes the $(m,n)$-th element of the matrix ${\bf X}$,  $\theta_{B/M}$  is the central angle of the outgoing/incoming rays from/to the $N_t$ transmit/$(N-N_t)$ receive antennas of the BS and $\sigma_{\theta}^{B/M}$ is the standard deviation of the corresponding angular spread. The comparisons between analytical and simulation results are shown in Figs. \ref{fig5:partial CSI} and \ref{fig6:partial CSI}.  In Fig. \ref{fig5:partial CSI}, we take $\sigma_{\theta}^{M}=10^{\circ}$, $\theta_M=15^{\circ}$, $\gamma_B=10$ bpcu,  and vary $P$.  In Fig. \ref{fig6:partial CSI}, $\sigma_{\theta}^{B}=10^{\circ}$, $\theta_B=5^{\circ}$, and $\sigma^2_{h_b}=30$ dBm are taken, and $\sigma^2_{h_m}$ is varied. A randomly selected unit beamformer is used  in both figures and $\alpha=0.1$ is chosen{\footnote{Note that the analytical and numerical results exhibit very good matching for any other beamformer and $\alpha$. For brevity, we show only a specific result.}}.
\begin{figure}[tb!]
  \centering
	\includegraphics[width=0.9\columnwidth]{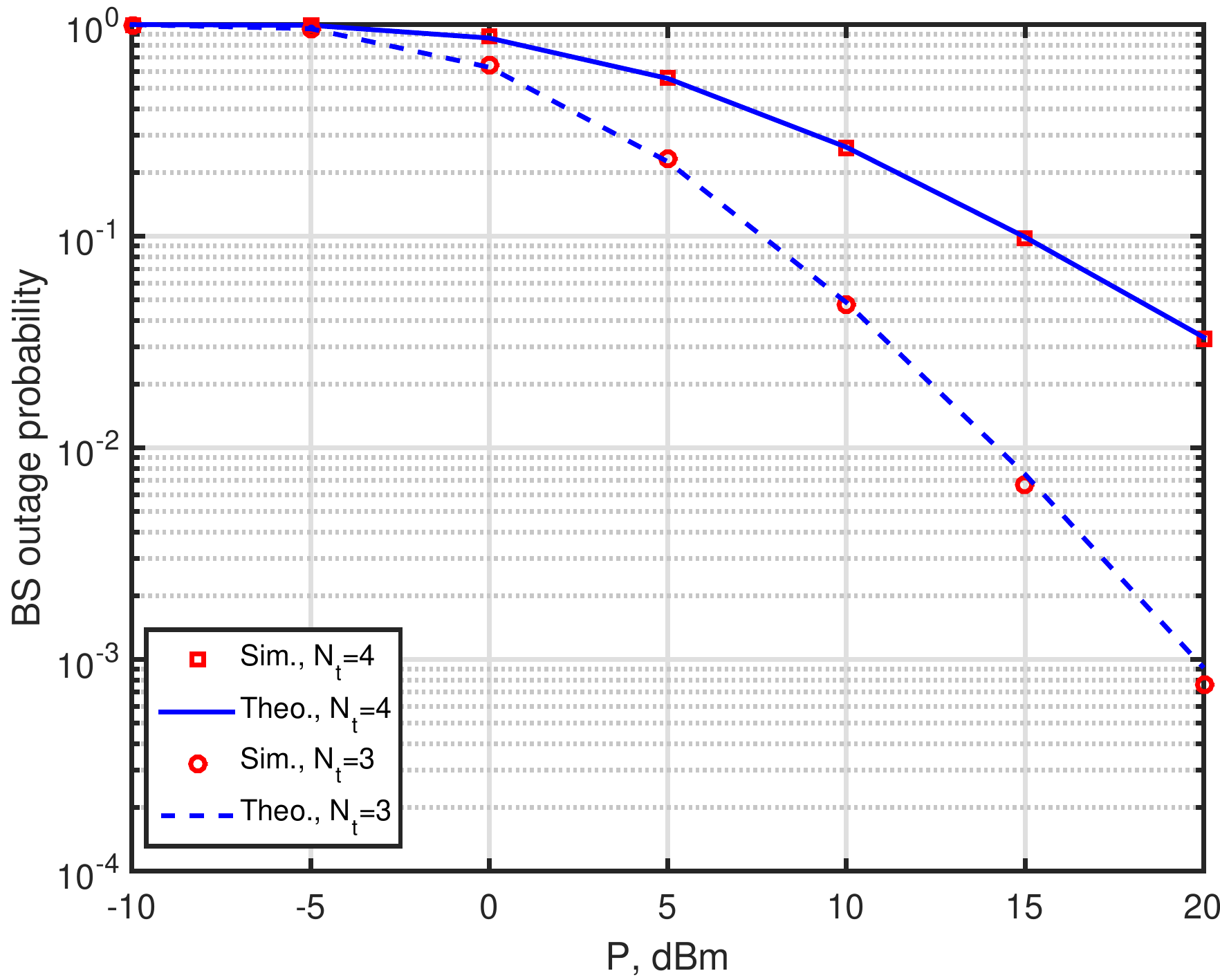}
	\vspace*{-0.3cm}
 \caption{Comparison of analytical and simulated BS outage probability.}
\label{fig5:partial CSI}
\end{figure}
It can be observed from Fig. \ref{fig5:partial CSI} that there is a good matching between the simulated and the analytical outage probability of the BS. Similarly, Fig. \ref{fig6:partial CSI} shows that the simulated and theoretical results of  the ergodic information rates at the MS exhibit also very good matching. Thus, these results verify the accuracy of the derived analytical expressions of outage probability and ergodic rate. 

\begin{figure}[tb!]
  \centering
	 \includegraphics[width=0.9\columnwidth]{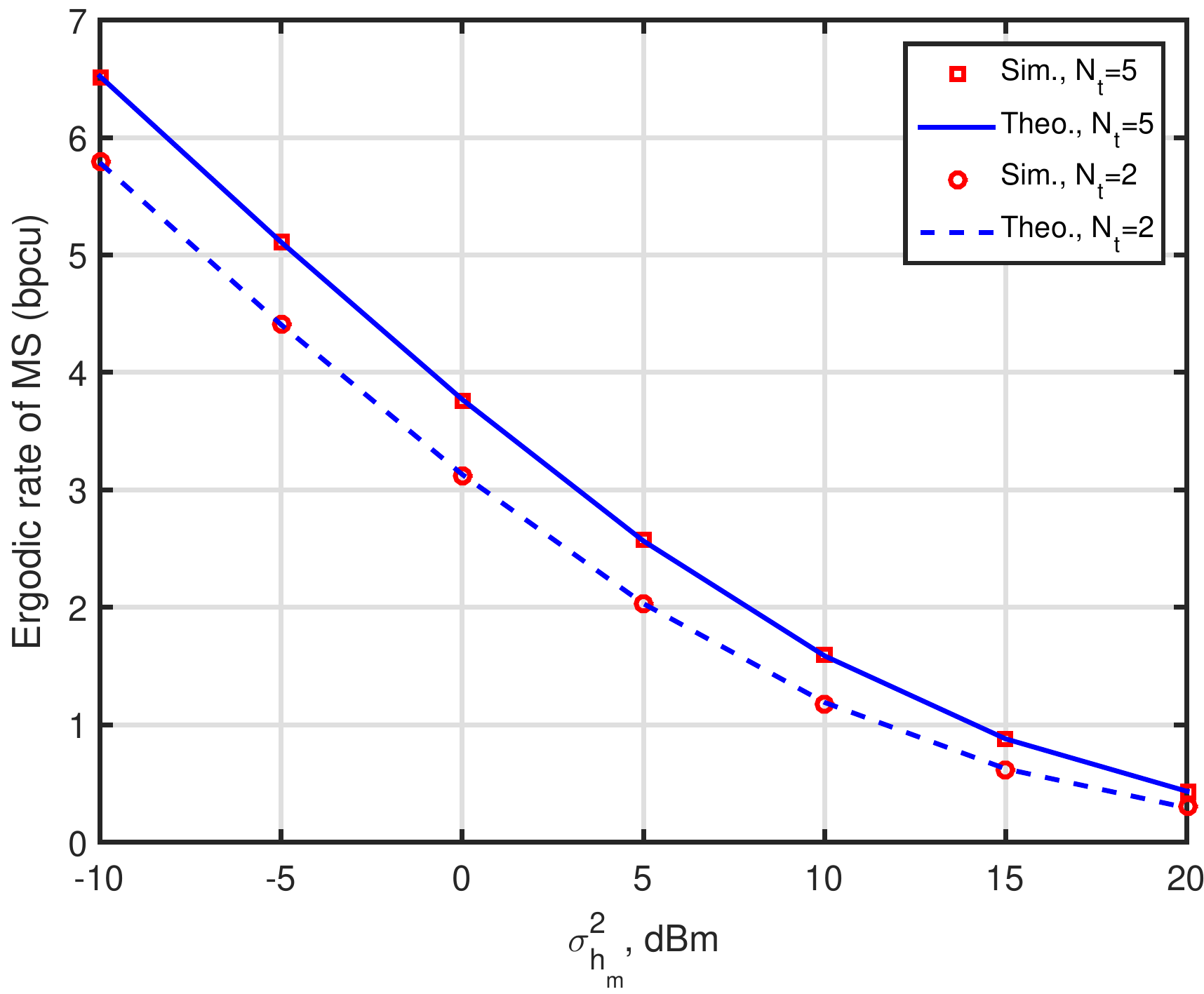}
	 \vspace*{-0.3cm}
 \caption{Ergodic MS rate versus variance of LI channel at the MS.}
\label{fig6:partial CSI}
\end{figure}

\begin{figure}[tb!]
  \centering
	 \includegraphics[width=0.9\columnwidth]{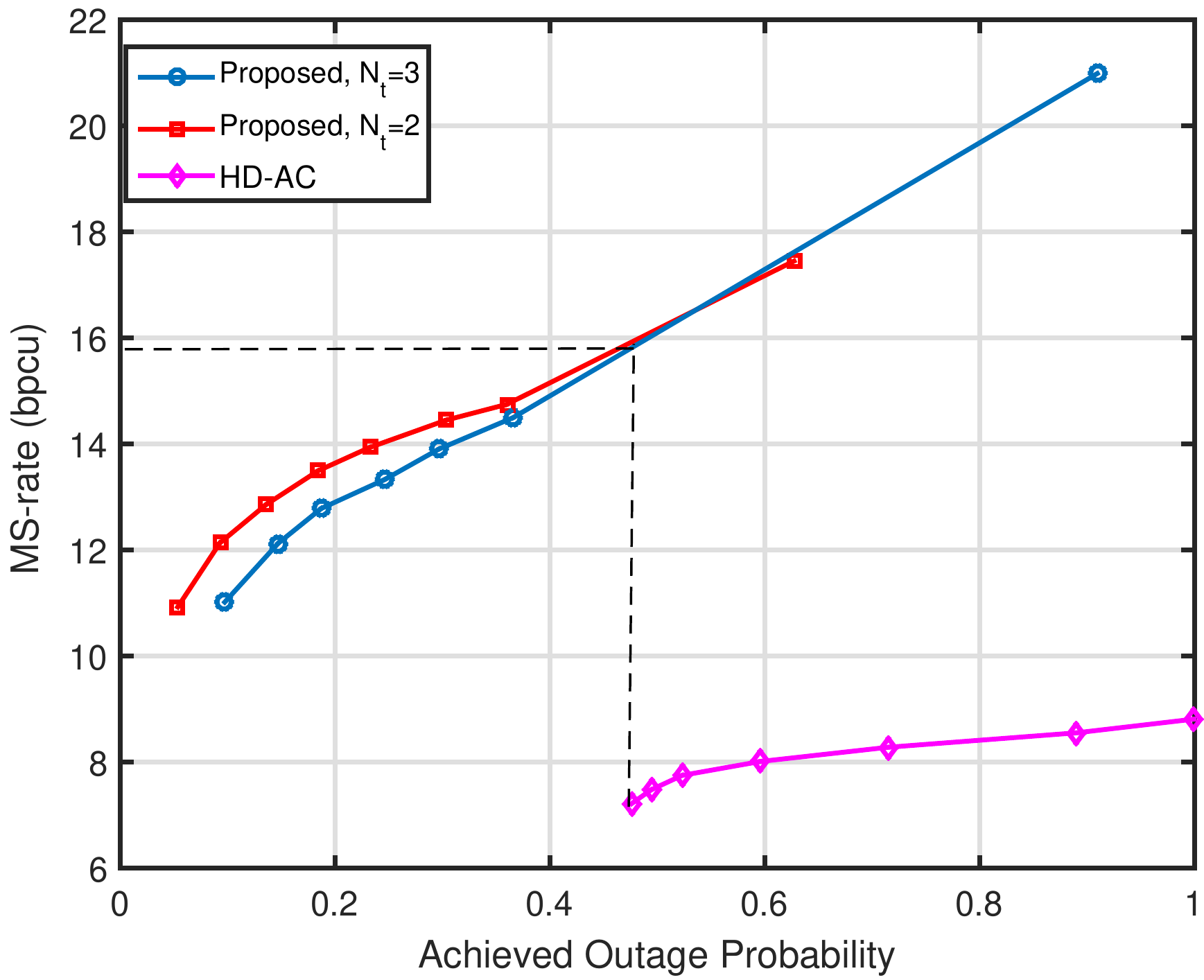}
	 \vspace*{-0.3cm}
 \caption{Ergodic MS rate versus BS outage probability for $N_t=2, N_t=3$, $N=6$, $P=10$ dBm, $\sigma^2_{h_b}=\sigma^2_{h_m}=30$ dBm, $\gamma_B=3$ bpcu.}
\label{fig7:partial CSI}
\end{figure}
The ergodic MS-rate versus BS outage probability is depicted in Figs. \ref{fig7:partial CSI} and  Figs. \ref{fig8:partial CSI} for different values of $N_t$. In both figures, we take $\sigma_{\theta}^{B}=\sigma_{\theta}^{M}=10^{\circ}$, $\theta_B=5^{\circ}$, $\theta_M=15^{\circ}$. In Fig. \ref{fig7:partial CSI}, $P=10$ dBm and $\gamma_B=3$ bpcu are taken, whereas in  Fig. \ref{fig8:partial CSI}, $P=0$ dBm and $\gamma_B=1$ bpcu are taken. The performance of the proposed FD scheme is compared with that of the HD-AC scheme in these figures{\footnote{ For conciseness, the performance of the HD-RFC scheme is skipped, since it is inferior to the performance of the HD-AC scheme.}}.

It can be observed from these figures that the MS-rate has to be sacrificed for achieving lower outage probability at the BS. Moreover, from the results of Figs. \ref{fig7:partial CSI} and  \ref{fig8:partial CSI}, we observe that the MS rate drops whereas the outage probability improves when $N_t$ decreases (or $N-N_t$ increases). This is due to the fact that smaller $N_t$ decreases beamforming gain of the BS towards the MS, whereas the resulting larger value of $N-N_t$ increases the LI suppression capability of the BS. Both figures also demonstrate that the proposed FD scheme significantly outperforms the benchmark HD-AC scheme, despite the fact that the former scheme requires only half of the RF chains than in the latter. 

For example, in Fig. \ref{fig7:partial CSI}, the achieved minimum outage probability with the  HD-AC scheme is only about 0.47, whereas that achieved  with the proposed method (with $N_t=2$) is approx. $0.05$ (improvement by a factor of 10).  At the outage probability of 0.47,  the HD-AC method achieves the MS-rate of only 7.2 bpcu, whereas the corresponding rate in the proposed scheme is about 10.9 bpcu. Similarly, in Fig. \ref{fig8:partial CSI},  the proposed scheme achieves the minimum outage probability of 0.017 (with $N_t=2$), whereas the HD-AC scheme achieves only 0.06. At this outage probability, the rate of the HD is 4.28 bpcu, whereas the corresponding rate of the proposed method is 7.05 bpcu. 

\begin{figure}[tb!]
  \centering
	\includegraphics[width=0.9\columnwidth]{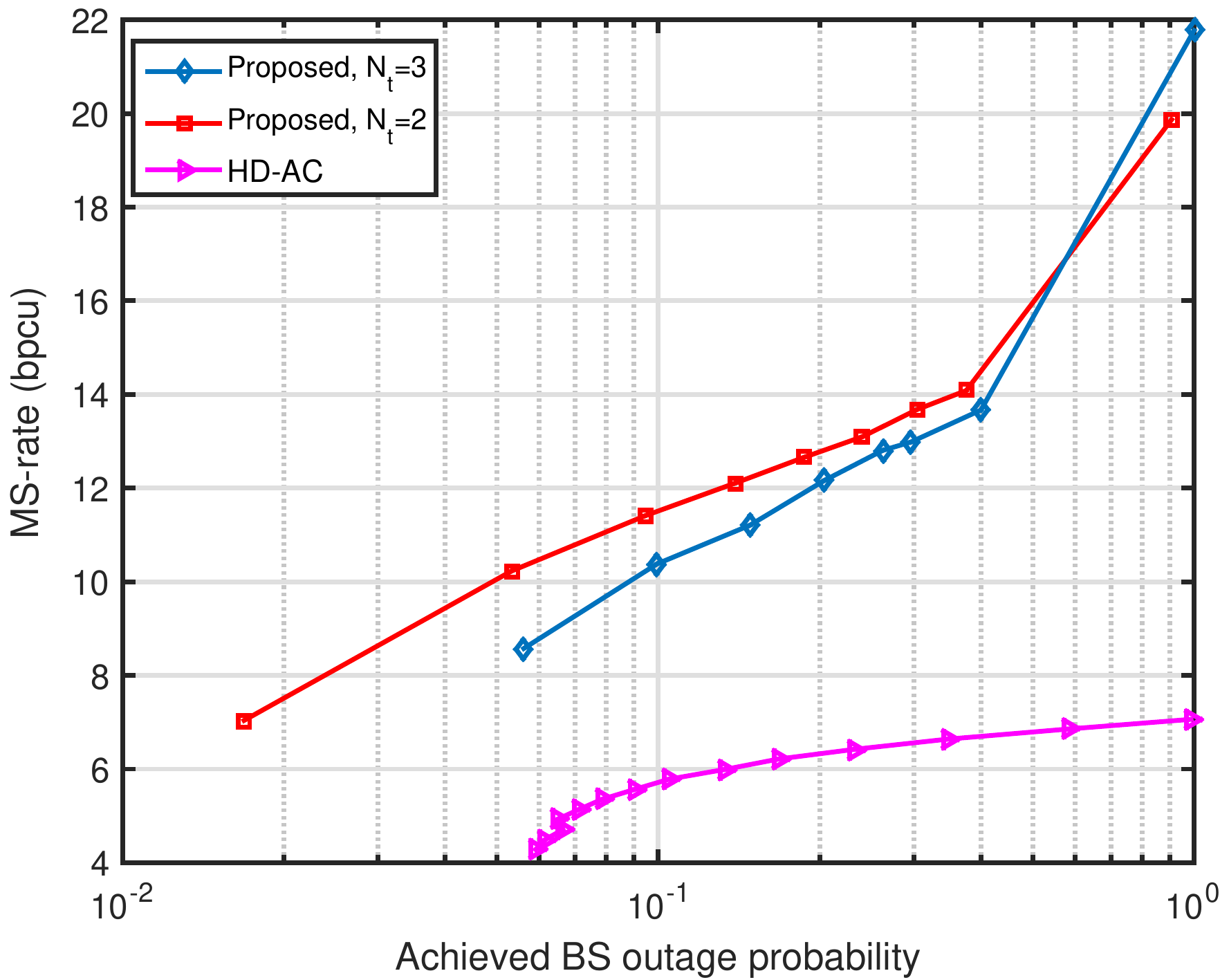}
	\vspace*{-0.3cm}
 \caption{Ergodic MS rate versus BS outage probability for $N_t=2, 3$, $N=6$, $P=0$ dBm, $\sigma^2_{h_b}=\sigma^2_{h_m}=30$ dBm, $\gamma_B=1$ bpcu.}
\label{fig8:partial CSI}
\end{figure}

\section{Conclusions}
\label{sec:Section5}
In this paper, the joint optimization of transmit beamforming and TS parameter was considered for a wireless-powered  bidirectional FD communication system. When instantaneous CSI is available, the boundary of  the MS-BS information rate was obtained by efficiently solving the optimization problem as an SDR problem,  in which the optimality of the relaxation was analytically confirmed. A sub-optimum approach based on zero-forcing constraint was also proposed,  where a closed-form expression for TS parameter was also determined. When the BS and MS have only second-order statistics of their transmit CSI, the joint optimization  was formulated as a problem of maximizing the ergodic MS rate, while satisfying the constraint on the BS outage probability. Utilizing the monotonicity property of the ergodic MS-rate and an upper bound of the outage probability, an SDR-based optimization problem was formulated and efficiently solved. Simulations demonstrate that significant performance gains are achievable over the half-duplex scheme when the beamformer and the TS parameter are jointly optimized.

\section*{Appendix A}
\subsection*{\centering{Derivation of $R_B^{\rm max}$}}
\label{Appendix:ProofPropositionBS-maxR}
It is obvious that 
\bea
\label{eqn:Ap1eq1}
\|\qh_M\|^2  - \frac{|\qh_M^H\qH_B\qw_B |^2  }{\sigma_b^2+ \|\qH_B\qw_B\|^2} \leq \|\qh_M\|^2,
\eea
where the equality is achieved with the ZF constraint $\qh_M^H\qH_B\qw_B =0$. The maximum BS rate is then obtained as
 \bea
\label{eqn:Ap1eq2}
R_B^{\rm max}=\max_{0<\alpha<1} (1-\alpha)\log_2\left( 1+\frac{\alpha}{1-\alpha} b \frac{\|\qh_M\|^2}{\sigma_b^2}\right), 
\eea
where $b=\eta P {\bar \lambda}\left(\qH_{BM}\qH_{BM}^H\right)$. Denote ${\tilde b}=b\frac{\|\qh_M\|^2}{\sigma_b^2}$. Equating the first order derivative of $ R_B^{\rm max}$ w.r.t. $\alpha$, we obtain
 \bea
\label{eqn:Ap1eq3}
\frac{\partial R_B^{\rm max}}{\partial \alpha}=0\Longrightarrow \log\left( 1+\frac{\alpha}{1-\alpha}{\tilde b}\right)=\frac{{\tilde b}}{  1+\frac{\alpha}{1-\alpha}{\tilde b} }\frac{1}{1-\alpha},
\eea
which can be written in the form
 \bea
\label{eqn:Ap1eq4}
z\log(z)=z+{\tilde b}-1, ~{\rm where}~z=1+\frac{\alpha}{1-\alpha}{\tilde b}.
\eea
After some straightforward manipulations, we obtain
 \bea
\label{eqn:Ap1eq5}
\frac{z}{{\rm e}}\log\left(\frac{z}{\rm e}\right)=\frac{{\tilde b}-1}{{\rm e}}\Longrightarrow \log\left( \frac{z}{{\rm e}}\right) {\rm e}^{ \log\left( \frac{z}{{\rm e}}\right)}=\frac{{\tilde b}-1}{{\rm e}}.
\eea
According to the definition of Lambert-W function \cite{Knuth}, the solution of the equation $y=xe^{x}$ for a given $y$ is expressed as $x=W(y)$, where $W(\cdot)$ is the Lambert-W function. Thus, \eqref{eqn:Ap1eq5} is given by
 \bea
\label{eqn:Ap1eq6}
z={\rm e}^{ W\left( \frac{{\tilde b}-1}{{\rm e}}\right)+1}.
\eea
Substituting $z$ into (\ref{eqn:Ap1eq6}), the optimum $\alpha$ is
\bea
\label{eqn:Ap1eq7}
\alpha^{\rm Opt}=\frac{  {\rm e}^{ W\left( \frac{{\tilde b}-1}{{\rm e}}\right)+1} -1} {{\tilde b}+ {\rm e}^{ W\left( \frac{{\tilde b}-1}{{\rm e}}\right)+1} -1}.
\eea
Therefore, $R_B^{\rm max}$ is given by
\bea
\label{eqn:Ap1eq7}
R_B^{\rm max}=(1-\alpha^{\rm Opt})\log_2\left( 1+\frac{\alpha^{\rm Opt}}{1-\alpha^{\rm Opt}} \frac{b \|\qh_M\|^2}{\sigma_b^2}\right). 
\eea

\section*{Appendix B}
\subsection*{\centering{Proof of Proposition ~\ref{ProofPropositionRankOne}}}
\label{Appendix:ProofPropositionRankOne}
 The Lagrangian multiplier function for the optimization problem (\ref{eq:ProbForm4})  is
\bea
{\mathcal L  }( \qV_B, \qY, \lambda_1, \lambda_2)\hspace*{-0.3cm}&=\hspace*{-0.3cm}&  -{\rm tr}(\qV_B \qh_B \qh_B^H )-{\rm tr}(\qY\qV_B )\nonumber\\
\hspace*{-0.3cm}& \hspace*{-0.3cm}& +\lambda_1\left[ {\rm tr}\left (\qV_B \qH_B^H \left( \qh_M\qh_M^H-\Gamma_B\qI\right)\qH_B\right )\right]\nonumber\\
\hspace*{-0.3cm}& \hspace*{-0.3cm}& +\lambda_2 \left({\rm tr}(\qV_B ) -P\right)-\lambda_1\sigma_b^2\Gamma_B,
\eea
where $\qY\succeq 0$ is the dual-variable associated with the positive semidefinite constraint  $\qV_B\succeq 0$, and $\lambda_1\geq 0$ \& $\lambda_2\geq 0$ are the Lagrangian multiplier coefficients. 
Among all KKT conditions, some relevant conditions required for the proof are as follows.
\bea
\label{eq:KKT1}
\frac{d {\mathcal L  }  }{ d \qV_B } & = &  \biggl\{ -\qh_B \qh_B^H+\lambda_1\qH_B^H\left(\qh_M\qh_M^H-\Gamma_B\qI\right) \qH_B \biggr. \nonumber\\ \biggl.
                                                               & &  +\lambda_2\qI-{ \bf Y}\biggr\}={ \bf 0 }\\
  \label{eq:KKT2}                                                             
 {\rm tr}(\qY\qV_B )&= & 0\Longrightarrow \qY\qV_B={ \bf 0 }, ~\qV_B \succeq 0, \qY \succeq 0.
\eea
The  KKT condition  (\ref{eq:KKT2})  implies that the optimum $\qV_B$ must lie in the null-space of $ \qY$. This means that the rank of optimum $\qV_B$ is the nullity of  $ \qY$. Consequently, it is sufficient to show that the optimum $ \qY$ has a nullity of one. Note that $\Gamma_B\geq 0$ and
\bea
{ \bf Y}=-\qh_B \qh_B^H+\lambda_1 \left[ \qH_B^H \left( \qh_M\qh_M^H-\Gamma_B\qI\right)\qH_B\right]+\lambda_2\qI.
\eea
Define $\qZ\triangleq \lambda_1 \left[ \qH_B^H \left( \qh_M\qh_M^H-\Gamma_B\qI\right)\qH_B\right]+\lambda_2\qI$. We first claim that at the KKT optimality, $\qZ$ is a positive-definite (full- rank) matrix. This can be readily proved by the method of contradiction. Consider the cases where $\qZ$ has at least one non-positive eigenvalue. Then, from Weyl's inequalities for sum of eigenvalues of Hermitian matrices  \cite{Horn}, it is clear that  $\qZ-\qh_B \qh_B^H$ will have at least one negative eigenvalue. In other words, ${ \bf Y}$ turns to a indefinite matrix, which contradicts the fact that $\qY$ should be positive-semidefinite. Consequently, positive-semidefiniteness of $\qY$ can be confirmed only when $\qZ$ is positive-definite. Now, we can show that the nullity of ${ \bf Y}$ cannot be greater than 1 by contradiction. Assume that $\left\{ \qu_{y,q}, q=1,2\right\} \in {\mathcal Ns  }({ \bf Y}) $, where  ${\mathcal Ns  }({ \bf Y})$ denotes null-space of ${ \bf Y}$. Then,
\bea
{ \bf Y}\qu_{y,q} & =& \qZ\qu_{y,q} -\qh_B \qh_B^H\qu_{y,q}\nonumber\\
& & \Longrightarrow { \bf 0}=\qZ\qu_{y,q} -\qh_B \qh_B^H\qu_{y,q}\nonumber\\
& & \Longrightarrow \qu_{y,q}= \qZ^{ -1}\qh_B \qh_B^H\qu_{y,q}, \forall q, 
\eea
which shows that $\qu_{y,q}$ is an eigenvector of $\qZ^{ -1}\qh_B \qh_B^H$ corresponding to eigenvalue $1$. Since ${ \rm rank }(\qZ^{ -1}\qh_B \qh_B^H)=1$, it turns out that $q$ cannot take a value greater than $1$. This shows that the dimension of null space of $\qY$ is one, and therefore, the rank of $\qV_B$ is one. \hfill$\square$

\section*{Appendix C}
\subsection*{\centering{Proof of Proposition ~\ref{ProofPropositionOptalpha}}}
\label{Appendix:ProofPropositionOptalpha}
The equality constraint for the BS rate is expressed as
\bea
\label{eq:App2eq1}
\log\left(1+\frac{\alpha}{1-\alpha}b\gamma\right)=R_B\log(2)\left(\frac{\alpha}{1-\alpha}+1\right).
\eea
Define $y\triangleq 1+\frac{\alpha}{1-\alpha}b\gamma$. Then (\ref{eq:App2eq1}) can be expressed in terms of $y$ as
\bea
\label{eq:App2eq2}
y={\rm e}^{ \frac{R_B\log(2)}{b\gamma} y}{\rm e}^{ R_B\log(2)\left( 1-\frac{1}{b\gamma}\right)},
\eea
which after simple manipulation can be expressed as
\bea
\label{eq:App2eq3}
\left( -\frac{R_B\log(2)}{b\gamma}y \right) {\rm e}^{- \frac{R_B\log(2)}{b\gamma} y}\hspace*{-0.3cm}&=\hspace*{-0.3cm}& \left( -\frac{R_B\log(2)}{b\gamma} \right)\nonumber\\
\hspace*{-0.3cm}& \hspace*{-0.3cm}& \times {\rm e}^{ R_B\log(2)\left( 1-\frac{1}{b\gamma}\right)}.
\eea
Using the Lambert-W function  $W(y)$ (i.e., $y=x{\rm e}^{x} \rightarrow x=W(y)$), $y$ in \eqref{eq:App2eq3} is expressed as
\bea
\label{eq:App2eq4}
y=\frac{- b\gamma}{R_B\log(2)} W\left( -\frac{R_B\log(2)}{b\gamma} {\rm e}^{  R_B\log(2) \left( 1-\frac{1}{b\gamma}\right)} \right).
\eea
Note that $\frac{R_B\log(2)}{b\gamma} {\rm e}^{  R_B\log(2) \left( 1-\frac{1}{b\gamma}\right)}\leq \frac{1}{\rm e}$ is required to have a real value of $y$. If not, the equality constraint is not feasible for given $b$, $\gamma$, and $R_B$ where $R_B\leq R_B^{\rm max}$. 
Substituting $y$ in \eqref{eq:App2eq4}, we obtain
\bea
\label{eq:App2eq5}
\frac{\alpha}{1-\alpha}=\frac{- 1}{R_B\log(2)} W\left( -\frac{R_B\log(2)}{b\gamma} {\rm e}^{  R_B\log(2) \left( 1-\frac{1}{b\gamma}\right)} \right)-\frac{1}{b\gamma},\nonumber
\eea
which yields the optimum ${\bar \alpha}^{\rm Opt}$ given in (\ref{eq:ProbForm13}). \hfill$\square$

\section*{Appendix D}
\subsection*{\centering{Proof of Proposition~\ref{ProofPropositionExact_RateB}}}
\label{Appendix:ProofPropositionExact_RateB}
Let ${\bar z}=\qh_M^H\left(\sigma_b^2\qI+\qH_B\qw_B\qw_B^H\qH_B\right)^{-1}\qh_M$. Since $\qh_M=\qR_M^{\frac{1}{2}}\qh_{M,w}$, where we consider that the effect of the distance dependent attenuation is included in $\qR_M$.
\bea
\label{eq:AppNewProof1}
{\bar z}=\qh_{M,w}^H\underbrace{\qR_M^{\frac{1}{2}}\left(\sigma_b^2\qI+\qH_B\qw_B\qw_B^H\qH_B\right)^{-1}\qR_M^{\frac{1}{2}}}_{{\boldsymbol \Phi}}\qh_{M,w}.
\eea
Let the eigenvalue decomposition of ${\boldsymbol \Phi}$ be given by ${\boldsymbol \Phi}={\bf U}{\boldsymbol \Lambda} {\bf U}^H$, where ${\boldsymbol \Lambda}$ is the diagonal matrix of eigenvalues  of ${\boldsymbol \Phi}$, whereas ${\bf U}$ is the matrix of eigenvectors. Let $\{ \lambda_i \}_{i=1}^{L}$ be the non-zero eigenvalues of  ${\boldsymbol \Phi}$, where its rank is $L$. Then, (\ref{eq:AppNewProof1}) is further expressed as
\bea
\label{eq:AppNewProof2}
{\bar z}&=&\qh_{M,w}^H {\bf U}{\boldsymbol \Lambda} {\bf U}^H\qh_{M,w}={\tilde {\bf h}}_{M,w}^H {\boldsymbol \Lambda} {\tilde {\bf h}}_{M,w}\nonumber\\
&=&\sum_{i=1}^{L} \lambda_i |{\tilde h}_{M,w,i}|^2,
\eea
where ${\tilde {\bf h}}_{M,w}={\bf U}^H\qh_{M,w}$ and ${\tilde h}_{M,w,i}$ is the $i$th element of  ${\tilde {\bf h}}_{M,w}$. Since ${\bf U}$ is a unitary matrix, the elements of  ${\tilde {\bf h}}_{M,w}$ remain ZMCSCG as the elements of  ${\bf h}_{M,w}$. The outage probability at the BS is given by
\bea
\label{eq:AppNewProof3}
P_{out, B}= {\rm Pr}\left\{ (1-\alpha) \log_2( 1+p_m{\bar z})\leq \gamma_B  \right\},
\eea
which after applying (\ref{eq:AppNewProof2}) gives
\bea
\label{eq:AppNewProof4}
P_{out, B} &= & {\rm Pr}\left\{  \sum_{i=1}^{L}  \lambda_i |{\tilde h}_{M,w,i}|^2 \leq \frac{1}{p_m}\left[ 2^{\frac{\gamma_B}{1-\alpha}}-1\right] \right\}\nonumber\\
&=& {\rm Pr}\left\{  \sum_{i=1}^{L}  \lambda_i |{\tilde h}_{M,w,i}|^2 \leq {\bar \gamma} \right\}.
\eea
Since ${\tilde h}_{M,w,i}$ is ZMCSCG with unit variance, $|{\tilde h}_{M,w,i}|^2 $ is exponentially distributed with unit parameter. Then, the RV $X=\sum_{i=1}^{L}  \lambda_i |{\tilde h}_{M,w,i}|^2 $ is a weighted sum of independent exponentially distributed random variables. The PDF of $X$ is given by \cite[p.11]{Cox}
\begin{eqnarray}
\label{eq:AppNewProof5}
f_X(x)=\sum_{i=1}^{L}a_i{\rm e}^{\displaystyle{-\frac{x}{\lambda_i}}},
\end{eqnarray}
where $x\geq 0$, and for $L> 1$ 
\begin{eqnarray}
\label{eq:AppNewProof6}
a_i=\frac{ \lambda_i^{L-2}}{\prod_{j=1,j\neq i}^{L} \lambda_i-\lambda_j }.
\end{eqnarray}
For $L=1$,  $a_i$ takes the values of $\frac{1}{\lambda_i}$. Substituting the PDF of $X$ into (\ref{eq:AppNewProof4}), we get
\begin{eqnarray}
\label{eq:AppNewProof7}
P_{out, B} =\int_{0}^{\bar \gamma} f_X(x) \,dx & = & \sum_{i=1}^{L} a_i\int_{0}^{\bar \gamma} {\rm e}^{-\frac{x}{\lambda_i}}\, dx\nonumber\\
&=& \sum_{i=1}^{L} a_i\lambda_i\left[1 - {\rm e}^{-\frac{{\bar \gamma}}{\lambda_i}}\right]
\end{eqnarray}
which completes the proof. \hfill$\square$

\end{document}